\documentclass[reqno,10pts]{article}

\usepackage{endnotes}
\let\footnote=\endnote

\usepackage{natbib}
 \bibpunct[, ]{(}{)}{,}{a}{}{,}%
\usepackage{amsmath}
\usepackage{amsfonts}
\usepackage{amssymb}         %additional math symbols (see amsguide.tex)
\usepackage{amsopn}
\usepackage{latexsym}
\usepackage{graphicx}        %graphic and .eps files
\usepackage{mathrsfs}		%calligraphic fonts
\usepackage{psfrag}
\usepackage{algorithm}
\usepackage[noend]{algorithmic}
\usepackage{color}
\usepackage{verbatim}
\usepackage{url}
\usepackage{multirow}
\usepackage{afterpage}
\usepackage{marvosym}
\usepackage{enumerate}
\usepackage{mathbbol}
\usepackage{array}
\usepackage{qcircuit}
\usepackage{mathtools}
\usepackage{listings}
\usepackage{hyperref}
\usepackage[a4paper, total={6.8in, 9in}]{geometry}

\newtheorem{theorem}{Theorem}

\newtheorem{proposition}{Proposition}

\newenvironment{proof}
 {{\sl Proof.}\hspace*{1 ex}}%
 {{\nopagebreak\hspace*{\fill}$\Box$\par\vspace{12pt}}}

\newcommand{\floor}[1]{{\left\lfloor{#1}\right\rfloor}}
\newcommand{\ceil}[1]{{\left\lceil{#1}\right\rceil}}

\newcommand{\R}{\ensuremath{\mathbb{R}}}
\newcommand{\C}{\ensuremath{\mathbb{C}}}

\DeclarePairedDelimiter\bra{\langle}{\rvert}
\DeclarePairedDelimiter\ket{\lvert}{\rangle}
\DeclarePairedDelimiterX\braket[2]{\langle}{\rangle}{#1 \delimsize\vert #2}
\DeclarePairedDelimiterX\dotp[2]{\langle}{\rangle}{#1, #2}

\begin{document}
%\RUNAUTHOR{Nannicini}

%\RUNTITLE{Fast quantum subroutines for the simplex method}

\title{Fast quantum subroutines for the simplex method}

\author{Giacomo Nannicini\thanks{IBM Quantum, IBM T.J.~Watson research center, Yorktown Heights, NY 10598, \url{nannicini@us.ibm.com}} %, \URL{}}
} % end of the block

\date{\today}

%\KEYWORDS{Linear programming, Quantum algorithms, Simplex method} %\HISTORY{}

\maketitle

\abstract{%

  We propose quantum subroutines for the simplex method that avoid
  classical computation of the basis inverse. We show how to quantize
  all steps of the simplex algorithm, including checking optimality,
  unboundedness, and identifying a pivot (i.e., pricing the columns
  and performing the ratio test) according to Dantzig's rule or the
  steepest edge rule. The quantized subroutines obtain a polynomial
  speedup in the dimension of the problem, but have worse dependence
  on other numerical parameters. For example, for a problem with $m$
  constraints, $n$ variables, at most $d_c$ nonzero elements per
  column of the costraint matrix, at most $d$ nonzero elements per
  column or row of the basis, basis condition number $\kappa$, and
  optimality tolerance $\epsilon$, pricing can be performed in
  $\tilde{O}(\frac{1}{\epsilon}\kappa d \sqrt{n}(d_c n + d m))$ time,
  where the $\tilde{O}$ notation hides polylogarithmic factors;
  classically, pricing requires $O(d_c^{0.7} m^{1.9} + m^{2 + o(1)} +
  d_c n)$ time in the worst case using the fastest known algorithm for
  sparse matrix multiplication. For well-conditioned sparse problems
  the quantum subroutines scale better in $m$ and $n$, and may
  therefore have an advantage for very large problems. The running time
  of the quantum subroutines can be improved if the constraint matrix
  admits an efficient algorithmic description, or if quantum RAM is
  available.
  %% An
  %% important feature of our paper is that this asymptotic speedup does
  %% not depend on the data being available in some ``quantum form'': the
  %% input of our quantum subroutines is the natural classical
  %% description of the problem, and the output is the index of the
  %% variables that should leave or enter the basis.
}%

\section{Introduction}
\label{sec:intro}
The simplex method is one of the most impactful algorithms of the past
century; to this day, it is widely used in a variety of
applications. This paper studies some opportunities for quantum
computers to accelerate the simplex method. An extended abstract of
this work appeared in the proceedings of IPCO 2021
\citep{nanniqsimplexipco}.

The use of quantum computers for optimization is a central research
question that has attracted significant attention in recent years. It
is known that a quadratic speedup for unstructured search problems can
be obtained using Grover's algorithm \citep{grover96fast}. Thanks to
exponential speedups in the solution of linear systems
\citep{harrow2009quantum,childs2017quantum}, it seems natural to try to
translate those speedups into faster optimization algorithms, since
linear systems appear as a building block in many optimization
procedures. However, few results in this direction are known. A
possible reason for the paucity of results is the difficulty
encountered when applying a quantum algorithm to a problem whose data
is classically described, and a classical description of the solution
is required. We provide a simple example to illustrate these
difficulties.

Suppose we want to solve the system $A x = b$, where $A$ is an $m
\times m$ invertible matrix with at most $d$ nonzero elements per
column or row.  Using the quantum linear systems
algorithm of \citet{childs2017quantum}, the gate complexity of this
operation is $\tilde{O}(d \kappa \max\{T_{A}, T_b\})$, where $T_{A},
T_b$ indicate the gate complexity necessary to describe $A, b$ in a
certain input model, and $\kappa$ is the condition number of $A$. We
remark that here and in the rest of this paper, we measure the running
time for quantum subroutines as the number of basic gates (i.e., gate
complexity), as is usual in the literature. Notice that $m$ does not
appear in the running time, as the dependence is
polylogarithmic. However, we need $\tilde{O}(dm)$ gates to implement
$T_{A}$ for sparse $A$ in the gate model, and $\tilde{O}(m)$ gates are
necessary to implement $T_b$. This is natural for an accurate
representation of the input data, since $A$ has $O(dm)$ nonzero
elements and $b$ has $O(m)$ nonzero elements. If we also want to
extract the solution $x = A^{-1} b$ with precision $\delta$, using the
(optimal) tomography algorithm of \cite{nannitomography} we end up
with running time $\tilde{O}(\frac{1}{\delta} \kappa d^2 m^2)$. This
is slower than the time taken to classically compute an LU
decomposition of $A$, which is $O(d^{0.7} m^{1.9} + m^{2 + o(1)})$
\citep{yuster2005fast}. Thus, naive application of quantum linear
system algorithms (QLSAs) does not give any advantage.

Despite the aforementioned difficulties, a few fast quantum
optimization algorithms exist --- not necessarily based on QLSAs; here
we briefly discuss some representative examples, and defer a more
detailed comparison with our work to Sect.~\ref{sec:literature}.
\cite{brandao2017quantum,apeldoorn2017quantum} give polynomial
speedups for the solution of semidefinite programs and therefore also
linear programs (LPs). These two papers give a quantum version of
\citep{arora2016combinatorial}: while the algorithm is essentially the
same as its classical counterpart, the basic subroutines admit faster
quantum algorithms. The running time for LPs is
$\tilde{O}\left(\sqrt{mn}\left(\frac{Rr}{\epsilon}\right)^5\right)$,
where $R,r$ are bounds on the size of the optimal primal/dual
solution, and $\epsilon$ an optimality tolerance; this is faster than
any known classical algorithm when $mn \gg \frac{Rr}{\epsilon}$ ---
although as \cite{apeldoorn2017quantum} note, many natural SDP
formulations do not satisfy this requirement and in fact, $r$ and $R$
may depend on $n$. To achieve this speedup,
\cite{brandao2017quantum,apeldoorn2017quantum} assume that there
exists an efficient quantum subroutine to describe $A$ (i.e., in
polylogarithmic time), and output only a dual solution --- the primal
solution is encoded in a quantum state. If we insist on classical
input and output for the optimization problem, the overall running
time increases
significantly. \cite{kerenidis2018quantum,casares2019quantum} also
give polynomial speedups for LPs, using different variants of an
interior point algorithm. Specifically, \cite{kerenidis2018quantum}
give a running time of $\tilde{O}(\frac{n^{2}}{\delta^2} \kappa^3)$,
where $\delta$ is a feasibility tolerance (with a non-standard
definition of feasibility), and \cite{casares2019quantum} give a
running time of $\tilde{O}(\frac{1}{\epsilon^2} \kappa
\sqrt{n}(n+m)\bar{\|M\|}_F)$, where $\bar{\|M\|}_F$ is an upper bound
on the Frobenius norm of all intermediate matrices appearing during
the optimization. Both papers follow the classical algorithm, but
accelerate the basic subroutines performed at each iteration. To
achieve this speedup, these papers rely on QRAM, a form of quantum
storage. QRAM assumes that classical data can be accessed in
superposition, and this allows data preparation subroutines that are
exponentially faster than their equivalent implementation under the
standard gate model, in general. Assuming QRAM, the algorithms of
\cite{kerenidis2018quantum,casares2019quantum} have classical input
and output. Similar considerations apply to the work of
\citet{van2019quantum}, whose quantum algorithm for LPs
is based on computing the Nash equilibrium of a two-player zero-sum game.

Summarizing, there are few known examples of faster quantum
optimization algorithms, and all of them have strong assumptions on
the availability of efficient data preparation or data readout
subroutines. In particular, the quantum optimization algorithms of
\cite{brandao2017quantum,apeldoorn2017quantum,kerenidis2018quantum,casares2019quantum,van2019quantum}
have one of these two assumptions: (i) that having quantum
input/output is acceptable, ignoring the cost of a classical
translation, or (ii) that QRAM, a form of quantum storage whose
physical realizability is still the subject of debate, is
available. Both assumptions have the merit of leading to interesting
algorithmic developments, but it is still an open question to find
practical situations in which they are satisfied, particularly in the
context of traditional optimization applications. We remark that the
assumptions can be dropped and the algorithms can be implemented in
the standard gate model, but the running time increases. In this paper
we propose quantum subroutines that may yield asymptotic speedups even
without these two assumptions. A more detailed analysis of the
existing literature, together with a comparison with our results, is
given in Sect.~\ref{sec:literature}.

\paragraph{Our results.} For brevity, from now on we assume that the
reader is familiar with standard linear optimization terminology; we
refer to \citep{bertsimas} for a comprehensive treatment of LPs. The
simplex method aims to solve $\min c^{\top} x, \text{s.t.:} \, Ax = b,
x \ge 0$, where $A \in \R^{m \times n}$ with at most $d_c$ nonzero
elements per column. It keeps a basis, i.e., a set of linearly
independent columns of $A$, and repeatedly moves to a different basis
that defines a solution with better objective function value. As is
common in the literature, we use the term ``basis'' to refer to both
the set of columns, and the corresponding submatrix of $A$, depending
on context. We denote by $B$ the set of basic columns, $N$ the set of
nonbasic columns, with corresponding submatrices $A_B, A_N$. The
maximum number of nonzero elements in any column or row of the basis
submatrix is denoted $d$. The basis change (called a {\em pivot}) is
performed by determining a new column that should enter the basis, and
removing one column from the current basis. The operations performed
to determine a pivot also identify if the current basis is optimal, or
if the problem is unbounded from below. Assessing which columns can
enter the basis is called {\em pricing}, and it is asymptotically the
most expensive step: it requires computing the basis inverse and
looping over all the columns in the worst case, for a total of
$O(d_c^{0.7} m^{1.9} + m^{2 + o(1)} + d_c n)$ operations using the
matrix multiplication algorithm of \cite{yuster2005fast}. In practical
implementations, the basis inverse is computed from scratch every few
iterations (this is often a user parameter); in other iterations,
inexpensive updates are used instead, see e.g., \cite{chvatalbook},
with a more favorable worst-case running time of $O(m^2 + d_c
n)$. With the steepest edge pivoting rule, that can achieve better
performance on real-world problems by reducing the number of required
iterations \citep{forrest1992steepest}, the term $d_c n$ in the two
running time expressions above increases to $m^2n$. (In practical
implementations of the simplex method, steepest edge is typically used
ony for dual simplex algorithms, due to its high computational cost;
an approximate variant, with more efficient updates, is used in primal
simplex instead.)

In the following, we denote by $T_{\text{LS}}(L, R, \epsilon)$ the
running time of a QLSA on the linear system $Lx = r$ with precision
$\epsilon$, where $r$ is a column of $R$ specified by some index. Note
that this accounts for the cost of preparing a quantum state encoding
any column of $R$, which is important if we want the ability to solve
$Lx = r$ for all columns $r$ of $R$ in superposition. We denote
$c_{\max} := \max_j c_j$ and $\eta := \max\{\|b\|, \max_j\{\|A_j\|\}\}$.

We show that we can apply Grover search to choose an entering column,
so that the running time scales as $O(\sqrt{n})$ rather than $O(n)$
for looping over all the columns. To apply Grover search we need a
quantum oracle that determines if a column is eligible to enter the
basis, i.e., if it has negative reduced cost. We propose a
construction for this oracle using a QLSA, several gadgets to make
state amplitudes interfere in a certain way, and amplitude estimation
\citep{brassard2002quantum}. The construction avoids classical
computation of the basis inverse. The overall running time of the
oracle is $\tilde{O}(\frac{1}{\epsilon}T_{\text{LS}}(A_B, A_N,
\frac{\epsilon}{2}))$, where $\epsilon$ is the precision for the
reduced costs (i.e., the optimality tolerance). Using the QLSA of
\cite{childs2017quantum}, in the circuit model and without taking
advantage of the structure of $A$ besides sparsity, this gives a total
running time of $\tilde{O}(\frac{1}{\epsilon}\kappa d \sqrt{n}(d_c n +
d m))$. If the ratio $n/m$ is large, we can find a better
tradeoff between the Grover speedup and the data preparation
subroutines, and improve the running time of the quantum pricing
algorithm to $\tilde{O}(\frac{1}{\epsilon} \kappa d^{1.5} \sqrt{d_c} n
\sqrt{m})$. We can also apply the steepest edge pivoting rule
increasing the running time by a factor $c_{\max}$. We summarize this
below.
\begin{theorem}
  \label{thm:intro_pricing}
  There exist quantum subroutines to identify if a basis is optimal,
  or determine a column with negative reduced cost, with running time
  $\tilde{O}(\frac{1}{\epsilon}\sqrt{n}T_{\text{LS}}(A_B, A_N,
  \frac{\epsilon}{2}))$. In the gate model without QRAM, this is
  $\tilde{O}(\frac{1}{\epsilon}\kappa d \sqrt{n}(d_c n + d m))$, which
  can be reduced to $\tilde{O}(\frac{1}{\epsilon} \kappa d^{1.5}
  \sqrt{d_c} n \sqrt{m})$ if the ratio $n/m$ is larger than $2
  \frac{d}{d_c}$. If QRAM to store $A$ is available, the
  running time is $\tilde{O}(\frac{1}{\epsilon}\kappa
  \sqrt{mn})$. With the steepest edge pivoting rule, the running time
  of the subroutine to determine a column entering the basis increases
  by a factor $c_{\max}$.
\end{theorem}
The regime with large $n/m$ is interesting because it includes many
natural LP formulations; e.g., the LP relaxations of cutting stock
problems, vehicle routing problems, or any other formulation that is
generally solved by column generation
\citep{lubbecke2005selected}. The optimality tolerance for the quantum
pricing subroutine is slightly different from the classical algorithm:
we need a tolerance relative to some norm for efficiency. For steepest
edge, the increased running time reflects the precision necessary to
run an approximate quantum minimum finding algorithm; these details
are discussed subsequently in the paper. Note that the running
time of the quantum subroutines depends explicitly on the condition
number of the basis and the precision of reduced costs $\epsilon$ is
fixed, while classically $\kappa$ is not explicit when using Gaussian
elimination, but the $\epsilon$ obtained would depend on it (because
the basis inverse could be inaccurate). Since the algorithm may fail
if the condition number grows too large (classically, because the
calculations become too imprecise and we can no longer accurately
check optimality or feasibility; quantumly, because the running time
increases too much), we quantize a heuristic rule commonly used in
classical implementations of the simplex method to reduce ``bad
pivots'', i.e., basis changes that deteriorate the condition
number. If $A$ is structured, the quantum running time can decrease
significantly: for example, if $A$ differs from the assignment problem
constraint matrix only for a polylogarithmic number of elements, then
its description in the sparse oracle access model used in this paper
requires time $\tilde{O}(1)$, rather than $\tilde{O}(d_cn)$. Our
running time analysis assumes that the matrix is sparse but the
sparsity pattern is unstructured. An important remark is that all
quantum subroutines succeed with some constant probability, that can
be boosted to at least $1-\gamma$ with $O(\log \frac{1}{\gamma})$
repetitions; in our informal theorem statements we omit this
aspect for simplicity.

If pricing is performed via our quantum subroutine, we obtain the
index of a column that has negative reduced cost with arbitrarily high
probability. To determine which column should leave the basis, we have
to perform the {\em ratio test}; scanning the rows during the ratio
test also allows us to detect if the problem is unbounded. Using
techniques similar to those used for the pricing step, we can identify
the column that leaves the basis in time $\tilde{O}(\frac{1}{\delta}
\eta \kappa^2 d^2 m^{1.5})$, where $\delta$ is the precision parameter
of this step and quantifies the maximum infeasibility after the
pivot. Classically, the ratio test
requires time $O(m^2)$ in the worst case, because the basis inverse
could be dense even if the basis is sparse. We summarize this result
below.
\begin{theorem}
  \label{thm:intro_ratio_test}
  There exists a quantum subroutine to identify if a nonbasic column
  proves unboundedness of the LP in time
  $\tilde{O}(\frac{\eta \kappa \sqrt{m}}{\delta} T_{\text{LS}}(A_B, A_B,
  \delta))$. There also exists a quantum subroutine to perform the
  ratio test in time $\tilde{O}(\frac{\eta \kappa \sqrt{m}}{\delta}
  T_{\text{LS}}(A_B, A_B, \delta))$, where $\delta$ is the maximum
  infeasibility of the basic solution after pivoting. In the gate
  model without QRAM, the running times are respectively
  $\tilde{O}(\frac{1}{\delta} \eta \kappa^2 d^2 m^{1.5})$ and
  $\tilde{O}(\frac{1}{\delta} \eta \kappa^2 d^2 m^{1.5})$. If QRAM to store
  $A$ and $b$ is available, the running times are respectively
  $\tilde{O}(\frac{1}{\delta} \eta \kappa^2 m)$ and
  $\tilde{O}(\frac{1}{\delta} \eta \kappa^2 m)$.
\end{theorem}
The factor $\eta \kappa$ in the expressions above comes from requiring
an absolute feasibility tolerance $\delta$; if precision relative to
some column norms is sufficient, the factor disappears. The exact
statement of these theorems is given subsequently in the paper.

It is known that for most practical LPs the maximum number of
nonzeroes in a column is essentially constant; for example, on the
entire benchmark set MIPLIB2010, less than 1\% of the columns have
more than 200 nonzeroes (and less than 5\% have more than 50
nonzeroes). Similarly, the number of nonzeroes per row of the basis is
small: on MIPLIB2010, looking at the optimal bases of the LP
relaxations, less than $0.01\%$ of the rows have more than 50
nonzeroes. As $m, n$ increase, so typically does the sparsity. For
example, the largest problem available in the benchmark set MIPLIB2017
has $m \approx 7.1 \times 10^6, n \approx 3.9 \times 10^7$, and
$99.999\%$ of the columns have less than 30 nonzero elements; for the
second largest problem, which has $m \approx 2.0 \times 10^7, n
\approx 2.1 \times 10^7$, $99.998\%$ of the columns have this
property. Hence, although LPs may have ``global'' constraints leading
to dense columns or rows, we expect many bases arising from real-world
LPs to be extremely sparse, and it is interesting to look at the
scaling of the running time under the assumption that the sparsity
parameters are at most polylogarithmic in $m$ and $n$. In this case,
the running time of the oracle for the reduced costs in the gate model
without QRAM is $\tilde{O}(\frac{\kappa}{\epsilon}(n + m))$, giving a
total running time to choose an entering column of
$\tilde{O}(\frac{1}{\epsilon}\kappa \sqrt{n}(n + m))$, and the
steepest edge pivoting rule is a factor $c_{\max}$ slower. Hence, for
a well-conditioned basis and under the assumption (often verified in
practice) that $d = O(\log mn)$, we obtain running time
$\tilde{O}(\frac{1}{\epsilon}\sqrt{n}(n + m))$ for the quantum pricing
subroutine, which can be reduced to
$\tilde{O}(\frac{1}{\epsilon}n\sqrt{m})$ if the ratio $n/m$ is large;
and running time $\tilde{O}(\frac{1}{\delta} \eta m^{1.5})$ for the quantum
ratio test subroutine. With QRAM, the gate complexity decreases
further, and the proposed algorithms achieve essentially linear
scaling in $m$ and $n$: $\tilde{O}(\frac{1}{\epsilon} \sqrt{mn})$ and
$\tilde{O}(\frac{\eta}{\delta} m)$, respectively.

Summarizing, the quantum subroutines that we propose can be
asymptotically faster than the best known classical version of them,
under the assumption that the LPs are extremely sparse --- an
assumption that is generally verified in practice --- and the bases
are well-conditioned. In addition to the classical input parameters
($m$, $n$, and the sparsity parameters), the gate complexity of the
quantum subroutines depends on some numerical parameters: optimality
and feasibility tolerances, maximum norm $\eta$ of a column of $A$ or
$b$, and the maximum cost coefficient (for steepest edge). In practice
the tolerances are typically chosen independently of $m$ and $n$, and
they can be treated as constants, while the parameter $\eta$ depends
on the sparsity. For well-conditioned problems, the quantum
subroutines have better scaling in $m$ and $n$, and this could turn
into an asymptotic advantage. To achieve this potential advantage, we
never explicitly compute the basis inverse, and rely on the quantum
computer to indicate which columns should enter and leave the basis at
each iteration.  Similar to other papers in the quantum optimization
literature, we use a classical algorithm (the simplex method) and
accelerate the subroutines executed at each iteration of the simplex.
However, our asymptotic speedup (when the condition number and some
other numerical parameters are small) does not depend on the
availability of QRAM or of the data in ``quantum form''. The key
insight to obtain an asymptotic speedup even with classical input and
output is to interpret the simplex method as a collection of
subroutines that output only the changes in the basis, avoiding the
cost of performing full quantum state tomography (i.e., obtaining a
classical description of the quantum state). The main disadvantage of
our method is arguably the inverse dependence on the precision
parameters $\epsilon$ and $\delta$; as discussed in the next section,
existing quantum optimization algorithms methods in the literature
have worse dependence on the precision parameters (although the
precision parameters have a different meaning and are not directly
comparable). However, for very large problems, the quantum subroutine
may be able to perform low-precision pivots faster than a classical
algorithm.

Our algorithms require a fault-tolerant quantum computer, therefore
they are not suitable for implementation on the noisy
intermediate-scale quantum computers currently available. Some recent
studies \citep{sanders2020compilation,babbush2021focus} estimate the
running time of a few quantum algorithms (including optimization
algorithms) and compare it with their quadratically slower classical
counterparts. These studies conclude that with the current state of
quantum technologies, a small polynomial speedup of the quantum
algorithm over classical is entirely offset by the cost of performing
error correction. The algorithms discussed in our paper likely suffer
from the same weakness, and might be numerically advantageous only
with significant progress in fault-tolerant quantum computation.

We remark that even with sophisticated pivoting rules, the number of
iterations of the classical simplex method could be exponential. The
quantum version proposed in this paper does not circumvent this issue,
and its worst-case running time is slower than than of the fastest
quantum algorihtm for LPs \citep{van2019quantum}. It is well
established that in practice the simplex method performs much better
than the worst case, in terms of the number of iterations
\citep{spielman2004smoothed,dadush2018friendly} as well as the
complexity of a single iteration, see e.g.~\cite{chvatalbook}.  The
most attractive feature of the simplex method is its excellent
practical performance, and we hope that a quantized version, closely
mimicking the classical counterpart while accelerating the linear
algebra carried out at each iteration, would inherit this
trait.
%% Besides showing that quantum computers may accelerate the
%% simplex method on large, sparse problems, our results may lead to
%% further developments in quantum optimization algorithms. For example,
%% it seems plausible that our subroutines can be concatenated to perform
%% multiple pivots in superposition; this may have interesting
%% applications even if the probability of success decreases
%% exponentially in the number of consecutive pivots. These directions
%% are left for future research.

The rest of this paper is organized as follows. We start with a more
detailed review of existing literature in Sect.~\ref{sec:literature}.
Sect.~\ref{sec:simplex} contains a brief summary of the simplex method
to establish terminology. In Sect.~\ref{sec:qsimplex} we define our
notation, describe useful results from the literature, and give an
overview of our quantum subroutines. Sect.~\ref{sec:qsimplex_details}
gives a detailed explanation of each step in the gate model without
QRAM, while Sect.~\ref{sec:qram} explains how to modify the algorithm
to improve the running time if we assume that QRAM is
available. Sect.~\ref{sec:conclusions} concludes the paper.

\section{Comparison with the existing literature}
\label{sec:literature}
The simplex method has been extensively studied in the operations
research and computer science literature. It was known since the 70s
that depending on the pivoting rule, the algorithm may take
exponential time on certain inputs \citep{klee1972good}, which is in
contrast with its excellent practical performance. Attempts at
explaining this behavior led to the idea of ``smoothed analysis''
\citep{spielman2004smoothed}, see the recent work of
\cite{dadush2018friendly} for an overview. This paper proposes a
quantization of the simplex method with Dantzig's rule, the steepest
edge pivoting rule, or with a randomized rule that chooses uniformly
at random among the possible pivots. It is known that in expectation,
the ``random edge'' pivot rule may not explore all vertices of the
polytope, see \citep{gartner2007two} for a discussion; note that this
still leads to a potentially exponential-time
algorithm. Sub-exponential bounds in expectation are known for some
other randomized pivoting rules
\citep{kalai1992subexponential,matouvsek1996subexponential}. In
general, no polynomial bound is known so far in the worst case, but
for randomly distributed problem instances better average upper bounds
are possible \citep{borgwardt1982average}.

On the quantum side, to the best of our knowledge all algorithms for
LPs are derived from some classical algorithm. In the introduction we
mentioned the papers
\citep{brandao2017quantum,apeldoorn2017quantum,kerenidis2018quantum,casares2019quantum,van2019quantum}.
\citep{brandao2017quantum,apeldoorn2017quantum} are based on the
multiplicative weights update
method. \citep{kerenidis2018quantum,casares2019quantum} are based on
the interior point method. \citep{van2019quantum} is based on a
reduction of LPs to two-player zero-sum games and the classical
algorithm of \cite{grigoriadis1995sublinear}.  For these methods the
number of iterations is polynomial (or better), and is taken directly
from the classical algorithm; the computational complexity of each
iteration is reduced taking advantage of quantum subroutines.  A
faster quantum interior point algorithm is described in
\citep{augustino2021quantum}, whereas a faster version of the quantum
multiplicative weights update method for LPs is given in
\citep{van2018improvements}.

We summarize key features of several papers in Table \ref{tab:litsummary}. For our algorithm, we report the running time of the most general (and therefore slowest) version; multiple opportunities for speedups are discussed throughout the paper.
\begin{table}
  \centering
  \caption{Quantum algorithms for linear programming, and comparison with the classical simplex.}
  \label{tab:litsummary}
  \setlength\tabcolsep{3pt}
  \begin{tabular}{|p{1in}|p{1.8in}|c|c|p{1.4in}|}
    \hline
    Algorithm & Iteration cost & \# Iterations & QRAM size & Comments \\
    \hline
    Multiplicative weights update \citep{van2018improvements} & $\tilde{O}\left((\sqrt{m} + \sqrt{n}\frac{Rr}{\epsilon})d\left(\frac{Rr}{\epsilon}\right)^2\right)$ & $O(\frac{R^2 \log n}{\epsilon^2})$ & $\tilde{O}\left(\left(\frac{Rr}{\epsilon}\right)^2\right)$ & Outputs dual solution and (quantum) primal solution; $R,r$ could depend on $n,m$ \\
    \hline
    Interior point \citep{augustino2021quantum} & $\tilde{O}(\kappa^2 \frac{n^{2}}{\epsilon})$ & $O(\sqrt{n} \log \frac{n}{\epsilon})$ & $\tilde{O}(d_cn + m)$ & $\kappa$ comes from intermediate matrices\\
    \hline
    Game-theoretical \citep{van2019quantum} & $\tilde{O}\left(\sqrt{d}\left(\frac{Rr}{\epsilon}\right)^{1.5}\right)$ & $\tilde{O}\left(\left(\frac{Rr}{\epsilon}\right)^2\right)$ & $\tilde{O}\left(\left(\frac{Rr}{\epsilon}\right)^2\right)$ & $R,r$ could depend on $n,m$ \\
    \hline
    \hline
    This paper & $\tilde{O}(\frac{1}{\epsilon}\kappa d \sqrt{n}(d_c n +  d m))$, plus $\tilde{O}(\frac{1}{\delta} \eta \kappa^2 d^2 m^{1.5})$ & N/A (exp) & No QRAM &  \multirow{2}{1.4in}{$\kappa$ comes from the\\current basis; outputs only basis information} \\
    \cline{1-4}
    This paper & $\tilde{O}(\frac{1}{\epsilon}\kappa \sqrt{mn})$, plus $\tilde{O}(\frac{1}{\delta} \eta \kappa^2 m)$ & N/A (exp) & $\tilde{O}(d_cn + m)$ & \\
    \hline
    \hline
    Classical simplex & $O(d_c^{0.7}m^{1.9} + m^{2+o(1)} + d_cn)$ & N/A (exp) & No QRAM & \\
    \hline
  \end{tabular}
\end{table}
The table highlights the main advantages of our method, in particular
the fact that the iteration running time is a polynomial with very low
degree even without QRAM: for all other methods, each iteration is
significantly more expensive, and even more so if we consider the gate
complexity in case QRAM is not available (i.e, the gate complexity
increases by a factor equal to the size of the QRAM, also reported in
the table). For some methods, the steep dependence on $\epsilon$ could
be a limiting factor. If QRAM is available, our subroutines have the
very attractive feature of essentially linear dependence on $m$ or $n$
(i.e., $O(\sqrt{mn})$ or in $O(m)$). On the other hand, the
methods proposed in this paper suffer from the same weakness as the
classical simplex method: giving a sub-exponential upper bound on the
number of iterations is difficult (we mark the upper bounds as ``N/A
(exponential)'' in the table). However, this has not prevented the
simplex method from being extremely efficient in
practice. Furthermore, due to the dependence on $\kappa$, the proposed
subroutines can be faster than classical only for well-conditioned
problems.

As remarked in the Table \ref{tab:litsummary}, the papers using the
multiplicative weights update framework as well as
\citep{van2019quantum} have a running time that depends on a parameter
$\frac{Rr}{\epsilon}$, which may in turn depend on $n, m$ --- see the
discussion in \cite{apeldoorn2017quantum}, as well as the application
to experimental design discussed in \citep{van2018improvements} to
understand the necessary tradeoffs to obtain a quantum speedup in $n$
and $m$. Note that $\epsilon$-optimality of the reduced cost, as used
in the simplex method, is not a global optimality guarantee, and
therefore our $\epsilon$ parameter is not directly comparable to the
$\epsilon$ used in the other algorithms discussed in Table
\ref{tab:litsummary}: we discuss this in more detail in
Sect.~\ref{sec:qsimplex_overview}. We also remark that in our paper,
the running time depends on the condition number $\kappa$ of the basis
at each iteration of the simplex method. It is known that $\kappa$ of
the intermediate bases could be worse than that of the initial basis
--- this is true also for the classical simplex method. We discuss
some steps to numerically keep the condition number under control in
Sect.~\ref{sec:harris}, similarly to existing implementations of the
classical simplex. An analogous depence on $\kappa$ can be found in
\citep{kerenidis2018quantum,augustino2021quantum}, but in those papers
it is exacerbated by the fact that $\kappa$ of the Newton linear
system for interior point methods grows as we approach optimality,
whereas for the simplex method, $\kappa$ comes from submatrices of the
initial constraint matrix and does not necessarily grow
large. Regarding input and output, the algorithm presented in this
paper has fully classical input, and outputs the current basis at each
iteration; to obtain the (primal) solution, it is necessary to
classically solve a single linear system of size $m \times m$ (this is
more efficient than obtaining a solution via a QLSA).

\section{Overview of the simplex method}
\label{sec:simplex}
The simplex method solves the following linear optimization problem:
$\min c^{\top} x, \text{s.t.:} Ax = b, x \ge 0$, where $A \in \R^{m
  \times n}$, $c \in \R^n, b \in \R^m$. A basis is a subset of $m$
linearly independent columns of $A$. Given a basis $B$, assume that it
is an ordered set and let $B(j)$ be the $j$-th element of the set. The
set $N := \{1,\dots,n\} \setminus B$ is called the set of nonbasic
variables. We denote by $A_B$ the square invertible submatrix of $A$
corresponding to columns in $B$, and $A_N$ the remaining
submatrix. The term ``basis'' may refer to the set of column indices
$B$ or to the submatrix $A_B$, depending on context.  The simplex
method can be described compactly as follows; see, e.g.,
\cite{bertsimas} for a more detailed treatment.
\begin{itemize}
\item Start with any basic feasible solution. (This is
  w.l.o.g.\ because it is always possible to find one.) Let $B$ be the
  current basis, $N$ the nonbasic variables, $x = A_B^{-1} b$ the
  current solution.
\item Repeat the following steps:
  \begin{enumerate}
  \item Compute the reduced costs for the nonbasic variables
    $\bar{c}_N^{\top} = c_N^{\top} - c_B^{\top} A_B^{-1} A_N$. This step
    is called {\em pricing}. If $\bar{c}_N \ge 0$ the basis is
    optimal: the algorithm terminates. Otherwise, choose $k :
    \bar{c}_k < 0$. Column $k$ is the {\em pivot column}.
  \item Compute $u = A_B^{-1} A_k$. If $u \le 0$, the optimal cost is
    unbounded from below: the algorithm terminates.
  \item If some component of $u$ is positive, compute (this step is
    called {\em ratio test}):
    \begin{equation}
      \label{eq:ratio_test}
      r^* := \min_{j=1,\dots,m : u_j > 0} \frac{x_{B(j)}}{u_j}.
    \end{equation}
  \item Let $\ell$ be such that $r^* = \frac{x_{B(\ell)}}{u_\ell}$.
    Row $\ell$ is the {\em pivot row}.  Form a new basis replacing
    $B(\ell)$ with $k$. This step is called a {\em pivot}. Update $x =
    A_B^{-1} b$.
    %% The new basic feasible solution $y = x - r^* u$ is
    %% such that $y_k = r^*$ and $y_{B(j)} = x_{B(j)} - r^*
    %% u_j$ for $j \neq \ell$.
  \end{enumerate}
\end{itemize}
To perform the pricing step, the standard classical approach is to
compute an LU factorization of the basis $A_B$; this requires time
$O(d_c^{0.7} m^{1.9} + m^{2 + o(1)})$ using fast sparse matrix
multiplication techniques \citep{yuster2005fast}. (In practice, the
traditional $O(m^3)$ Gaussian elimination is used instead, but the
factorization is not computed from scratch at every iteration.) Then,
we can compute the vector $c_B^{\top} A_B^{-1}$ and finally perform
the $O(n)$ calculations $c_k^{\top} - c_B^{\top} A_B^{-1} A_k$ for all
$k \in N$; this requires an additional $O(d_c n)$ time, bringing the
total time to $O(d_c^{0.7} m^{1.9} + m^{2 + o(1)} + d_c n)$.  To
perform the ratio test, we need the vector $u = A_B^{-1} A_k$, which
takes time $O(m^2)$ assuming the LU factorization of $A_B$ is
available from pricing. As remarked earlier, $A_B^{-1}$ or the LU
factors may not be dense if $A_B$ is sparse; furthermore, the vectors
$c_B^{\top} A_B^{-1}$ and $A_B^{-1} b$ can be updated from previous
iterations exploiting the factorization update, so all these step
could take significantly less time in practice. Finally, since the
calculations are performed with finite precision, we use an optimality
tolerance $\epsilon$ and the optimality criterion becomes $\bar{c}_N
\ge -\epsilon$.

It is well known that the performance of the simplex method in
practice depends on the pivoting rule. One of the simplest rules is
Dantzig's rule, which choooses $k = \arg \min_h \bar{c}_h$. Modern
implementations of the simplex method typically rely on more
sophisticated pivoting rules. Among these, the steepest edge pivoting
rule has been shown to lead to a significant reduction in the number
of iterations \citep{forrest1992steepest}, but the per-iteration cost
increases. With steepest edge pivoting, the choice of the column
entering the basis becomes:
\begin{equation*}
  k = \arg \min_h \frac{\bar{c}_h}{\|A_B^{-1} A_h\|}.
\end{equation*}
This is more expensive to compute than Dantzig's rule, as it requires
knowledge of the norms $\|A_B^{-1} A_h\|$. \cite{forrest1992steepest}
show how to update these norms in time $O(m^2 n)$, which
asymptotically is essentially the same as recomputing them from
scratch, provided the basis inverse is available (in practice,
sparsity of the basis inverse may make the updates easier). We remark
that the same paper introduces dual steepest edge pivoting rules, for
which the norm updates cost $O(m^3)$ instead, but in this paper we
focus our attention on the primal simplex.

\section{Quantum implementation: overview}
\label{sec:qsimplex}
Before giving an overview of our methodology, we introduce some
notation and useful results from the literature. The state of a
quantum computer with $q$ qubits is a unit vector in $(\C^2)^{\otimes
  q} = \C^{2^q}$; we denote the standard basis vectors by $\ket{j}$,
where $j \in \{0,1\}^q$ (e.g., when $q=2$, $\ket{01}$ denotes the
standard basis vector $(0, 1, 0, 0)^{\top}$).  Their conjugate
transpose is $\bra{j} = (\ket{j})^{\dag}$.  A quantum state is
therefore of the form $\ket{\psi} = \sum_{j \in \{0,1\}^q} \alpha_{j}
\ket{j}$, $\sum_{j \in \{0,1\}^q} |\alpha_{j}|^2 = 1$.  The final
state is obtained by applying a unitary matrix $U \in \C^{2^q \times
  2^q}$ to the initial state $\ket{0^q}$, where $0^q$ denotes the
$q$-digit all-zero binary string (i.e., the first standard basis
vector). We assume that the reader is familiar with quantum computing
notation; an introduction for non-specialists is given in
\citep{nanni2017introduction}, and a comprehensive reference is
\citep{nielsen02quantum}. Given a vector $v$, $\|v\|$ denotes its
$\ell_2$-norm; given a matrix $A$, $\|A\|$ denotes the spectral norm,
whereas $\|A\|_F$ is the Frobenius norm. Given two matrices $C, D$
(including vectors or scalars), $(C, D)$ denotes the matrix obtained
stacking $C$ on top of $D$, assuming that the dimensions are
compatible.  Given a classical vector $v \in \R^{2^q}$, we denote
$\ket{v} := \sum_{j=0}^{2^q-1} \frac{v_j}{\|v\|} \ket{j}$ its
amplitude encoding. If we have binary digits or strings $a_1, a_2$, we
denote by $a_1 \cdot a_2$ their concatenation. Given a binary string
$a = a_1 \cdot a_2 \cdot \dots \cdot a_m$ with $a_j \in \{0,1\}$, we
define $0.a := \sum_{j=1}^m a_j 2^{-j}$. The symbol $\oplus$ denotes
bitwise addition modulo 2, i.e., binary XOR. $\mathbf{1}_m$ is the
all-one vector of size $m$.

\subsection{Preliminaries}
\label{sec:preliminaries}
We define the following symbols:
\begin{itemize}
\item $d_c$: maximum number of nonzero entries in any column of $A$.
\item $d_r$: maximum number of nonzero entries in any row of $A_B$.
\item $d := \max\{d_c, d_r\}$: sparsity of $A_B$.
\item $\kappa$: ratio of largest to smallest nonzero singular value of
  $A_B$. Throughout this paper, we assume that $\kappa$, or an upper
  bound on it, is known. Note that since $\kappa$ is provided as an
  input to the QLSA, replacing it with an upper bound increases the
  running time by a corresponding amount.
\item $\eta := \max\{\max_{j} \|A_j\|, \|b\|\}$: maximum norm of a
  column of $A$ or $b$.
\item $L$: maximum nonzero entry of $A_B$, rounded up to a power of
  2. 
%\item $\varepsilon$: precision of the solution, i.e., the QLSA
%  returns $\ket{\tilde{x}}$ with $\|\ket{A_B^{-1}b} -
%  \ket{\tilde{x}}\| \le \varepsilon$.
\item $r := (\ceil{\log m}+\ceil{\log d})$: number of bits to index
  entries of $A_B$.
\end{itemize}
The notation $\tilde{O}$ is used to suppress polylogarithmic factors
in the input parameters, i.e., $\tilde{O}(f(x)) = O(f(x) \allowbreak
\text{poly}(\log n,\allowbreak \log m, \log \frac{1}{\epsilon}, \log \kappa,
\allowbreak \log d, \log L))$. Note that $L$ is exponential in the
input size, but we include $\log L$ among the factors suppressed in
$\tilde{O}$ notation because it simplifies the analysis: keeping track
of the number of bits used to represent entries is cumbersome and
largely ininfluential. As stated in the introduction, we assess the
complexity of quantum algorithms in terms of basic gates. We assume
that the cost of a controlled unitary is the same as that of the
unitary, because the number of additional gates required by the
controlled unitary is generally the same in the $\tilde{O}$ notation.

We use several common building blocks for quantum algorithms: phase
estimation, amplitude amplification \cite{grover96fast}, amplitude
estimation \citep{brassard2002quantum}. We state a version of these
building blocks below.
\begin{theorem}[Sect.~5.2 in \cite{nielsen02quantum}]
  \label{thm:qpe}
  Let $U$ be a unitary and $\ket{\psi}$ an eigenvector of $U$ with
  eigenvalue $2\pi i \theta$. Let $\varepsilon = 2^{-q}$ and $\gamma >
  0$. With probability at least $1-\gamma$, the phase estimation
  algorithm determines a $q$-qubit approximation $\tilde{\theta}$ of
  $\theta$, i.e., $|\theta - \tilde{\theta}| < 2^{-q} = \varepsilon$, using
  $O(2^q + \log \frac{1}{\gamma}) = O(\frac{1}{\varepsilon} + \log
  \frac{1}{\gamma})$ applications of $U$ and $O(\log
  \frac{1}{\varepsilon} \log \frac{1}{\gamma})$ aditional gates. In
  particular, if we let $0.a$ be the output of the procedure and we
  use $q + \ceil{\log(2 + \frac{1}{2\gamma})}$ qubits of precision,
  the first $q$ bits of $a$ are accurate with probability at least
  $1 - \gamma$, i.e., $\text{Pr}(|\theta - \sum_{j=1}^q a_j 2^{-j}| <
  2^{-q}) \ge 1 - \gamma$.
\end{theorem}
\begin{theorem}[\cite{grover96fast,brassard2002quantum}]
  \label{thm:ampamp}
  Let $U$ be a $q$-qubit unitary such that $U\ket{0^q} = \sqrt{p}
  \ket{\psi_{good}} + \sqrt{1-p} \ket{\psi_{bad}}$, where for some $G
  \subset \{0,1\}^q$, we have $\ket{\psi_{good}} = \frac{1}{\sqrt{p}}
  \sum_{j \in G} \alpha_j \ket{j}$, $\ket{\psi_{bad}} =
  \frac{1}{\sqrt{1-p}}\sum_{j \not\in G} \alpha_j \ket{j}$ and $p =
  \sum_{j \in G} |\alpha_j|^2$. Let $F$ be a unitary that maps
  $\ket{\psi_{good}} \to -\ket{\psi_{good}}$, $\ket{\psi_{bad}} \to
  \ket{\psi_{bad}}$. The amplitude amplification algorithm outputs $j
  \in G$ with probability at least 2/3 using
  $O(\frac{1}{\sqrt{p}})$ applications of $U$ and $F$, and additional
  gates. In particular, if we have a quantum oracle implementing a
  Boolean function $f : \{0,1\}^q \to \{0,1\}$ and let $G = \{x \in
  \{0,1\}^q : f(x) = 1\}$, we can determine an element of $G$ with
  $O\left(\sqrt{\frac{2^q}{|G|}}\right)$ applications of $U$, if $|G|$
  is known; if $|G|$ is unknown, we can determine an element of $G$
  with $O\left(\sqrt{2^q}\right)$ applications in expectation.
\end{theorem}
\begin{theorem}[\cite{brassard2002quantum}, Sect.~6.3 in \cite{nielsen02quantum}]
  \label{thm:ampest}
  Let $U, p, \ket{\psi_{good}}, \ket{\psi_{bad}}, F, G$ be defined as
  in Thm.~\ref{thm:ampamp}. Let $\sqrt{p} = \sin \theta$, with $0 \le
  \theta \le \frac{\pi}{2}$. Let $R$ be a unitary that maps $\ket{0^q}
  \to -\ket{0^q}$ and $\ket{j} \to \ket{j}$ for all other $j$. Then
  the state $U\ket{0^q}$ is an equally weighted superposition of
  eigenvectors of the operator $Q = U^{\dag} R U F$, with eigenvalues
  $e^{i2\theta}, e^{-i2\theta}$ respectively. In particular, applying
  phase estimation to the operator $Q$ and the state $U\ket{0^q}$,
  with $q$ qubits of precision, yields an estimate $\tilde{\theta}$
  such that $|\sin^2 \tilde{\theta} - p| \le 2 \pi
  \frac{\sqrt{p(1-p)}}{2^q} + \frac{\pi^2}{2^{2q}}$ with probability
  at least $2/3$. If we simply want to determine if $G$ is empty or
  not, $\ceil{q/2} + 3$ qubits of precision suffice, yielding an
  algorithm that succeeds with probability at least $2/3$ and requires
  $O(\sqrt{2^q})$ calls to $U$.
\end{theorem}
In several parts of this paper we use amplitude amplification on
oracles with bounded error: it has been shown that this is possible
while achieving the same running time as with deterministic oracles
\citep{hoyer2003quantum}. We also use an approximate version of
quantum minimum finding; we state its complexity below. Its
correctness follows from the result of
\citet{durr1996quantum}. Detailed proofs of this result, and all other
results in this paper, are available in the appendix.
\begin{theorem}[\cite{durr1996quantum}]
  \label{thm:minfind}
  Let $S := \{0,\dots,2^{q}-1\}$ and $f : S \to \R$. Suppose we have
  access to a unitary $U_f$ such that given $\ket{j}, \ket{k}$ with
  $j, k \in S$, the unitary $U_f$ returns $1$ if $(f(j) \le f(k) -
  \epsilon) \lor (j = k)$, and $0$ otherwise. There exists a quantum
  algorithm that finds $y \in S$ such that $f(y) \le \min_{x \in S}
  f(x) + \epsilon$ using at most $\frac{45}{4} \sqrt{2^q} +
  \frac{7}{10} q^2 = \tilde{O}(\sqrt{2^q})$ calls to $U_f$ in
  expectation.
\end{theorem}
We remark the (well-known) fact that statements about expected running
time, such as in Thm.s~\ref{thm:ampamp} and \ref{thm:minfind}, imply
that the right answer can be obtained with probability at least $2/3$
by executing the algorithm for three times the expected running time
(by Markov's inequality). We frequently use this alternative
characterization of the algorithms.

Finally, we discuss the QLSA introduced in \citep{childs2017quantum},
see also \cite{chakraborty2018power,gilyen2019quantum}. For the system
$A_B x = b$ with integer entries, the input to the algorithm is
encoded by two unitaries, $P_{A_B}$ and $P_b$, which are queried as
oracles defined as follows:
\begin{itemize}
\item $P_{A_B}$: specified by two maps; the map $\ket{j,\ell} \to
  \ket{j,\nu(j,\ell)}$ provides the index $\nu(j,\ell)$ of the
  $\ell$-th possibly nonzero element of column $j$, the map
  $\ket{j,k,z} \to \ket{j,k,z \oplus (A_{B})_{jk}}$ provides the value of
  the $k$-th element of column $j$. (This is the typical sparse
  encoding for vectors or matrices in the classical world as well.)
\item $P_b$: maps $\ket{0^{\ceil{\log m}}} \to \ket{b}$.
\end{itemize}
\begin{theorem}[Thm.~5 in \citep{childs2017quantum}]
  \label{thm:qls}
  Let $A_B$ be such that $\|A_B\| \le 1$. Given $P_{A_B}$, $P_b$,
  and $\varepsilon > 0$, there exists a quantum algorithm that
  produces the state $\ket{\tilde{x}}$ with $\|\ket{A_B^{-1}b} -
  \ket{\tilde{x}}\| \le \varepsilon$ using $\tilde{O}(d \kappa)$
  queries to $P_{A_B}$ and $P_b$, with additional gate complexity
  $\tilde{O}(d \kappa)$.
\end{theorem}
Note that \citep{childs2017quantum} discusses the case of a symmetric
matrix $A_B$, but it is known that this restriction can be relaxed by
considering the system $\begin{pmatrix} 0 & A_B \\ A_B^{\top} &
  0 \end{pmatrix} \begin{pmatrix} 0 \\ x \end{pmatrix}
= \begin{pmatrix} b \\ 0 \end{pmatrix}$, see
\cite{harrow2009quantum}. If the original (non-symmetric) $A_B$ is
invertible, so is the symmetrized matrix, and the singular values are
the same (but with different multiplicity). In the rest of this paper,
we refer to $A_B$ as a shorthand for the above symmetrized matrix; the
r.h.s.\ and the number of rows $m$ are adjusted accordingly. Our
running time analysis takes into account this transformation, i.e., we
use $d$ rather than $d_c$ where appropriate. The last result concerns
estimation of the norm of the solution of a linear system; this is
useful because the QLSA outputs a normalized solution, but sometimes
we need the norm of the unnormalized vector as well. We state a
simplified version, as the details are cumbersome and unimportant in
the context of our paper.
\begin{theorem}[Corollary 32 in \cite{chakraborty2018power}]
  \label{thm:qlsnormest}
  In the setting of Thm.~\ref{thm:qls}, we can output 
  $\tilde{e} \in \R$ such that $(1-\xi)\|A_B^{-1}\ket{b}\| \le \tilde{e} \le
  (1+\xi)\|A_B^{-1}\ket{b}\|$ using $\tilde{O}(\frac{\kappa d}{\xi})$
  queries to $P_{A_B}$ and $P_b$.
\end{theorem}
Thm.~\ref{thm:qlsnormest} immediately yields an estimate of
$A_B^{-1}b$ (without the normalization for $b$) with the same error
provided that $\|b\|$ is known; since this is always the case in this
paper (we can assume that the norm of $b$ and of the columns of $A$ is
computed in a preprocessing step), in the remainder we do not worry
about the r.h.s.~normalization in the context of
Thm.~\ref{thm:qlsnormest}.

%% Finally, we make extensive use of amplitude
%% estimation \citep{brassard2002quantum}. In most situations we use the
%% basic version of the algorithm, but in one specific case we employ a
%% non-destructive variant because it is slightly more efficient. It is
%% stated below.
%% \begin{theorem}[\citep{arunachalam2020simpler} Corollary 3.2]\label{thm:nondampest}
%% Let $\ket{\psi}$ be an arbitrary quantum state and $P$ an arbitrary
%% projector. Let $R_{\psi}=2\ket{\psi}\bra{\psi}-I$. Let $\epsilon\in
%% (0,1)$. Then, there is a quantum algorithm ${\cal A}$ that starts in
%% the initial state $\ket{\psi}$ and, with probability at least
%% $1-\eta$, outputs an estimate $\hat{p}$ of $p=\bra{\psi}P\ket{\psi}$
%% with additive error $\epsilon$. Additionally, ${\cal A}$ restores the
%% state $\ket{\psi}$ with probability $1-\eta$, and invokes the
%% controlled reflection $R_{\psi}$ $O\left( (1/\epsilon) \log( 1/\eta)
%% \right)$ many times.
%% \end{theorem}
%% The above theorem can be easily used to estimate the amplitude of a
%% given basis state, i.e., some coefficient $\alpha_j$ of a state
%% $\ket{\psi} = \sum_{j \in \{0,1\}^q} \alpha_{j} \ket{j}$, by choosing
%% the projector $P = \ket{j}\bra{j}$; if we have access to a unitary $U$
%% such that $U\ket{0^q} = \ket{\psi}$, then the construction of the
%% reflection $R_{\psi}$ is well-known, see, e.g.,
%% \cite{brassard2002quantum}.

A quantum RAM (QRAM) is a quantum-accessible form of storage that
allows for querying a superposition of addresses. Given a QRAM that
stores the classical vector $v_j \in \R^{2^q}$, and a quantum state
$\sum_{j=0}^{2^q-1} \alpha_j \ket{j}$, we can perform the following
mapping in time $\tilde{O}(1)$:
\begin{equation*}
  \sum_{j=0}^{2^q-1} \alpha_j \ket{j} \otimes \ket{0} \to \sum_{j=0}^{2^q-1} (\alpha_j \ket{j} \otimes \ket{v_j}).
\end{equation*}
Note that this mapping has gate complexity $\tilde{O}(2^q)$ in the standard
gate model without QRAM, because it would have to be constructed as a
``lookup table'' in the general case. Therefore, if an algorithm has
gate complexity ${\cal C}$ assuming access to a QRAM of size $2^q$,
there exists a quantum algorithm for the same task that does not
require QRAM access and has gate complexity $\tilde{O}(2^q{\cal C})$.
We should note that the physical realizability of QRAM is still an
open question. While there have been proposals for its construction,
such as \citep{giovannetti2008quantum}, the community has not reached a
consensus regarding the practical feasibility of these
proposals. Nonetheless, QRAM is used in virtually all papers on
quantum optimization, because the cost of describing classical data is
often prohibitive in the gate model without QRAM, negating any quantum
advantage.

\subsection{High-level description of the quantum subroutines}
\label{sec:qsimplex_overview}
As stated in the introduction, a naive application of a QLSA (with
explicit classical input and output) in the context of the simplex
method does not give any advantage over the classical algorithm. To
gain an edge, we need a different approach. The idea of this paper is
based on the observation that an iteration of the simplex method can
be reduced to a sequence of subroutines that have scalar
output. Indeed, the simplex method does not require explicit knowledge
of the full solution vector $A_B^{-1} b$ associated with a basis, or
of the full simplex tableau $A_B^{-1} A_N$, provided that we are able
to:
\begin{itemize}
\item Identify if the current basis is optimal or unbounded;
\item Identify a pivot, i.e., the index of a column with negative
  reduced cost that enters the basis, and the index of a column
  leaving the basis.
\end{itemize}
While subroutines to perform these tasks require access to $A_B^{-1}
b$ and/or $A_B^{-1} A_N$, we show that we can get an asymptotic
speedup by never computing a classical description of $A_B^{-1},
A_B^{-1} b$ or $A_B^{-1} A_N$. This is because extracting from the
quantum state the classical description of an amplitude-encoded vector
(an operation called ``state tomography'') is much more expensive than
obtaining scalar outputs, as these can be digitally encoded as basis
states and read from a single measurement with some probability.

Throughout this overview of the quantum algorithms, we assume that the
LP data is properly normalized. The normalization can be carried out
as a classical preprocessing step, whose details are given in
Sect.~\ref{sec:qsimplex_details}, and the running time of this step is
negligible compared to the remaining subroutines.

Our first objective is to implement a quantum oracle that determines
if a column has negative reduced cost, so that we can apply Grover
search to this oracle. To reach this goal we rely on the QLSA of
Thm.~\ref{thm:qls}. Using that algorithm, with straightforward data
preparation we can construct an oracle that, given a column index $k$
in a certain register, outputs $\ket{A_B^{-1}A_k}$ in another
register. We still need to get around three obstacles: (i) the output
of the QLSA is a renormalization of the solution, rather than the
(unscaled) vector $A_B^{-1}A_k$; (ii) we want to compute $c_k -
c_B^{\top} A_B^{-1}A_k$, while so far we only have access to
$\ket{A_B^{-1}A_k}$; (iii) we need the output to be a binary yes/no
condition (i.e., not just some information encoded in an amplitude) so
that Grover search can be applied to it. We overcome the first two
obstacles by: extending and properly scaling the linear system so that
$c_k$ is suitably encoded in the QLSA output; and using the inverse of
the unitary that maps $\ket{0^{\ceil{ \log m + 1}}}$ to
$\ket{(-c_B,1)}$ to encode $c_k - c_B^{\top} A_B^{-1}A_k$ in the
amplitude of one of the basis states. To determine the sign of such
amplitude, we rely on interference to create a basis state with
amplitude $\alpha$ such that $|\alpha| \ge \frac{1}{2}$ if and only if
$c_k - c_B^{\top} A_B^{-1}A_k \ge 0$. At this point, we apply
amplitude estimation (Thm.~\ref{thm:ampest}) for $O(1/\epsilon)$
iterations to test the condition with precision $\epsilon$. We
therefore obtain a unitary operation that overcomes the three
obstacles. A similar scheme can be used to determine if the basis is
optimal, i.e., no column with negative reduced cost exists.

Some further complications merit discussion. The first one concerns
the optimality tolerance: classically, this is typically $\bar{c}_N
\ge -\epsilon$ for some given $\epsilon$. Note that an absolute
optimality tolerance is not invariant to rescaling of $c$. In the
quantum subroutines, to ensure that the precision of the amplitude
estimation part is not too demanding, we use the inequality $\bar{c}_k
\ge -\epsilon\|(A_B^{-1}A_k,c_k)\|$ as optimality criterion, and show
that precision $O(\epsilon)$ is sufficient for this task. We remark
that if $\|A_B^{-1}A_k\|$ is small, the ratio test
\eqref{eq:ratio_test} has a small denominator, hence the basic
feasible solution $A_B^{-1}b$ may move by a large amount if column $k$
enters the basis; it is therefore reasonable to require that
$\bar{c}_k$ is very close to zero to declare that the basis is
optimal. The converse is true if $\|A_B^{-1}A_k\|$ is large. Thus, our
approach can be interpreted as some relative optimality criterion. As
is typical in the simplex method, the stopping criterion is based on
optimality of the reduced costs; even if this does not give guarantees
on the gap with respect to the optimal solution of the problem, this
is the method used in practical implementations of the
algorithm. Since we never explicitly compute the basis inverse, we do
not have classical access to $\|(A_B^{-1}A_k,c_k)\|$.  To alleviate
this issue, we show in Sect.~\ref{sec:normest} that there exists a
quantum subroutine to compute the root mean square of
$\|A_B^{-1}A_k\|$ over all $k$; thus, if desired we can also output
this number to provide some characterization of the optimality
tolerance used in the pricing step ($\|A_B^{-1}A_k\|$ for a specific
$k$ can be obtained with a direct application of
Thm.~\ref{thm:qlsnormest}.) It is important to remark that while the
optimality tolerance is relative to a norm, which in turn depends on
the problem data, the evidence provided in Sect.~\ref{sec:intro}
suggests that due to sparsity, in practice we expect these norms to
grow very slowly with $m$ and $n$ (i.e., polylogarithmically). To
implement Dantzig's pivoting rule and determine the smallest reduced
cost, as opposed to a random column with negative reduced cost, the
running time increases by a factor $c_{\max}$, because we increase
precision to perform pairwise comparisons of reduced costs in quantum
minimum finding. To implement the steepest edge pivoting rule the
running time is the same as for Dantzig's rule, but the analysis is a
bit more involved, as we need to properly renormalize the encoding of
the reduced cost to account for denominators of the form $\|A_B^{-1}
A_k\|$. For this, we use Thm.~\ref{thm:qlsnormest} and very simple
block-encoding techniques \citep{gilyen2019quantum}. Note that in the
classical case, steepest edge is more expensive than Dantzig's rule,
increasing the running time expression by an additive term $m^2n$.

%% checking $\bar{c}_k < -\epsilon$ for a column $k$
%% may be too expensive (i.e., $O(L/\epsilon)$) if the norm of
%% $\|(A_B^{-1}A_k,c_k)\|$ is too large; this is due to the normalization
%% of the quantum state, which would force us to increase the precision
%% of amplitude estimation. We therefore

The second complication concerns the condition number of the basis:
the running time of the quantum routines explicitly depends on it, but
this is not the case for the classical algorithms based on an LU
decomposition of $A_B$ (although precision may be affected). We do not
take any specific steps to improve the worst-case condition number of
the basis (e.g., \cite{clader2013preconditioned}), but we note that
similar issues affect the classical simplex method: even if the
running time does not depend on $\kappa$, when $\kappa$ grows very
large the algorithm may fail because the computation of primal
solutions or pivots becomes too imprecise. Many approaches have been
proposed to prevent this from happening in practice, e.g.,
modifications of the pivoting rule to select pivot elements that are
not too small, such as the two-pass Harris ratio test. We discuss in
Sect.~\ref{sec:harris} how to quantize the Harris procedure; we
suspect that other similar procedures can be quantized as well.

%% The
%% QLSA framework may also allow for ignoring singular values above a
%% certain cutoff, and inverting only the part of the basis that is
%% ``well-conditioned'', i.e., spanned by singular vectors with singular
%% value below the cutoff: this is explicitly discussed in
%% \citep[Sect.~IV]{harrow2009quantum}, and may also be possible in other
%% frameworks that rely at least partially on quantum phase estimation
%% (e.g., the variable-time amplitude amplification approach of
%% \cite{childs2017quantum}); we did not investigate this direction and
%% leave it as an open question. We remark that ignoring parts of the
%% spectrum of $A$ could lead to unbounded error (i.e., if the r.h.s.\ of
%% some linear system is orthogonal to the subspace considered
%% \citep{harrow2009quantum}), but in practice this would only happen if
%% the problem is numerically very unstable.

With the above construction we have a quantum subroutine that
determines the index of a column with negative reduced cost, if one
exists. Such a column can enter the basis. To perform a basis update
we still need to determine the column leaving the basis: this is our
second objective. For this step we need knowledge of $A_B^{-1}
A_k$. Classically, this is straightforward because the basis inverse
is available, since it is necessary to compute reduced costs
anyway. With the above quantum subroutines the basis inverse is not
known, and in fact part of the benefit of the quantum subroutines
comes from always working with the original, sparse basis, rather than
its possibly dense inverse. Thus, we describe another quantum
algorithm that uses a QLSA as a subroutine, and identifies the element
of the basis that is an approximate minimizer of the ratio test
\eqref{eq:ratio_test}. Special care must be taken in this step,
because attaining the minimum of the ratio test is necessary to ensure
that the basic solution after the pivot is feasible (i.e., satisfies
the nonnegativity constraints $x \ge 0$). However, in the quantum
setting we are working with continuous amplitudes, and determining if
an amplitude is zero is impossible due to finite precision. Our
approach to give rigorous guarantees for this step relies on binary
search for the optimal value of $r$ in \eqref{eq:ratio_test}, i.e.,
the maximum ``step length'' when moving along an edge of the
polyhedron, while ensuring that the violation of the nonnegativity
constraints is at most $\delta$.  With similar techniques we can
determine if column $k$ proves unboundedness of the LP. Note that
because of the inexactness of the ratio test, we could pivot to
slightly infeasible solutions, i.e., we only guarantee $x \ge -\delta
\mathbf{1}_m$ for a given $\delta > 0$; to prevent this from breaking
the algorithm, we borrow two strategies from the classical
simplex. The first strategy is a strict zero step for basic variables
that are negative after a pivot: this is the same as in the classical
world, see, e.g., \citep{huangfu2013high}. The second strategy is a
periodic recomputation of the basis to eliminate infeasibilities: we
give a quantum subroutine to determine $\delta$-feasibility, and if
the basis is infeasibile, we can switch to Phase 1 of the simplex
method, possibly increasing precision, until we recover (approximate)
feasibility. (Recall that Phase 1 is equivalent to solving the problem
$\min \sum_{j=1}^m s_i, Ax + s = b, x \ge 0, s\ge 0$, which is always
feasible assuming $b \ge 0$, and can be used to determine a feasible
basis for the original problem if one exists.)

\section{Details of the quantum implementation}
\label{sec:qsimplex_details}
We now discuss data preparation and give the details of the quantum
subroutines. In this section we work with the standard gate model and
do not assume that we have access to QRAM. Some of the data
preparation routines become significantly faster with QRAM: this is
the subject of Sect.~\ref{sec:qram}, where we also indicate the
(small) algorithm modifications that are required to take
advantage of quantum-accessible storage.

The steps followed at every iteration are listed in
Alg.~\ref{alg:simplex_iter}; all the subroutines used in the algorithm
are fully detailed in subsequent sections. For brevity we do not
explicitly uncompute auxiliary registers, but it should always be
assumed that auxiliary registers are reset to their original state at
the end of a subroutine; this does not affect the running time
analysis. Furthermore, we are only concerned with giving constant
lower bounds on the probability of success of the subroutines; all
probabilities of success can be made arbitrarily large with standard
approaches and polylogarithmic cost increase.

\begin{algorithm}
  \caption{Run one iteration of the simplex method: {\sc
      SimplexIter}$(A, B, c, \epsilon, \delta)$.}
  \label{alg:simplex_iter}
    \begin{algorithmic}[1]
      \STATE {\bf Input:} Matrix $A$, basis $B$, cost vector $c$,
      precision parameters $\epsilon$, $\delta$, $t$.
      
      \STATE {\bf Output:} Flag ``optimal'', ``unbounded'', or a pair
      $(k, \ell)$ where $k$ is a nonbasic variable with negative
      reduced cost, $\ell$ is the basic variable that should leave the
      basis if $k$ enters.

      \STATE \label{al:normalize} Normalize $c$ so that $\|c_B\| =
      1$. Normalize $A$ so that $\|A_B\| \le 1$.

      \STATE Apply {\sc IsOptimal}$(A, B, \epsilon)$ to determine if
      the current basis is optimal. If so, return ``optimal''.

      \STATE Apply {\sc FindColumn}$(A, B, \epsilon)$ to determine a
      column with negative reduced cost. Let $k$ be the column index
      returned by the algorithm. $k$ is the pivot column.

      \STATE Apply {\sc IsUnbounded}$(A_B, A_k, \delta)$ to determine
      if the problem is unbounded. If so, return ``unbounded''.

      \STATE Apply {\sc FindRow}$(A_B, A_k, b, \delta)$ to
      determine the index $\ell$ of the row that minimizes the ratio
      test \eqref{eq:ratio_test}. $\ell$ is the pivot row. Update the
      basis $B \leftarrow (B \setminus \{B(\ell)\}) \cup
      \{k\}$. Return $(k, \ell)$.
    \end{algorithmic}
\end{algorithm}

All data normalization is performed within the subroutine {\sc
  SimplexIter}$(A, B, c, \epsilon, \delta)$, on line
\ref{al:normalize}: this prepares the input for all other
subroutines. Normalizing $c$ has cost $O(n)$, while finding the
leading singular value of $A_B$ has cost $O(\frac{1}{\epsilon'} md
\log m)$ using, e.g., the power method, with precision $\epsilon'$
\citep{kuczynski1992estimating}. The QLSA requires the eigenvalues of
$A_B$ to lie in $[-1, -1/\kappa] \cup [1/\kappa, 1]$: those outside
this set are discarded by the algorithm. We can choose $\epsilon'$ to
be some arbitrary constant, say, $10^{-4}$; rescale $A$ by
$(1-\epsilon')/\sigma_{\max}$, where $\sigma_{max}$ is the estimate
obtained with the power method; and inflate $\kappa$ by a factor
$1/(1-\epsilon')$. This ensures that the spectrum of the rescaled
$A_B$ satisfies the requirements. The overall running time is not
affected, because the time for the remaining subroutines dominates
$O(\frac{1}{\epsilon'} md \log m)$, as shown in the rest of this
section (e.g., simply loading $A_B$ and $A$ for the QLSA requires time
$\tilde{O}(d_c n + d m)$, see the proof of
Thm.~\ref{thm:findcolumn}).

\subsection{Implementation of the oracles $P_{A_B}$ and $P_b$}
\label{sec:oracles}
Recall that the quantum linear system algorithm of Thm.~\ref{thm:qls}
requires access to two quantum oracles describing the data of the
linear system. We now discuss their implementation, so as to compute
their gate complexity. Given computational basis states
$\ket{k},\ket{j}$, in this section we denote $\ket{k,j} := \ket{k}
\otimes \ket{j}$ for brevity. Recall that $\oplus$ denotes the binary
XOR; and, for readers not familiar with quantum computers, recall that
the $X$ gate is equivalent to a bit-flip (i.e., $X \ket{0} = \ket{1},
X\ket{1} = \ket{0}$). We start with the two maps necessary for
$P_{A_B}$.

Since the map $\ket{j,\ell} \to \ket{j,\nu(j,\ell)}$ is in-place, we
implement it as the sequence of mappings $\ket{j,\ell,0^{\ceil{\log
      m}}} \to \ket{j,\ell,\nu(j,\ell)} \to \ket{j,\ell, \ell \oplus
  \nu(j,\ell)} \to \ket{j,\nu(j,\ell),0^{\ceil{\log m}}}$.

To implement the first part $\ket{j,\ell,0^{\ceil{\log m}}} \to
\ket{j,\ell,\ell \oplus \nu(j,\ell)}$, for given $j$ and $\ell$ we use
$O(\log m)$ controlled single-qubit operations. Let $U$ be the
transformation that maps the basis state $\ket{\ell}$ into
$\ket{\nu(j,\ell)}$: this can be computed classically in $O(\log m)$
time, since it requires at most $\ceil{\log m}$ bit-flips. Then the
desired mapping requires at most $O(\log m)$ $r$-controlled $X$ gates,
which require $O(r)$ basic gates each plus polylogarithmic extra
space. The mapping $\ket{j,\ell,\ell \oplus \nu(j,\ell)} \to
\ket{j,\nu(j,\ell), \ell \oplus \nu(j,\ell)}$ is then easy to
construct, as it requires at most $\ceil{\log m}$ CX gates to compute
$\ell \oplus (\ell \oplus \nu(j,\ell)) = \nu(j,\ell)$ in the second
register (by taking the XOR of the third register with the second). We
then undo the computation in the third register. Thus, we obtain a
total of $O(r \log m)$ basic gates for each $j$ and $\ell$. There are
at most $d$ nonzero elements per column (recall the transformation to
make $A_B$ symmetric), yielding $\tilde{O}(d)$ basic gates to
construct the first map of $P_{A_B}$ for a given column $j$; since
there are $m$ columns, the overall gate complexity of this oracle is
$\tilde{O}(dm)$.

The implementation of $\ket{j,k,z} \to \ket{j,k,z \oplus A_{jk}}$ is
similar. Since $j,k \in \{0,\dots,m-1\}$ and the largest entry of $A$
is $L$, the map requires $O(r \log L)$ basic gates for each $j$ and
$k$. Summarizing, the implementation of the two maps for $P_{A_B}$
requires $\tilde{O}(d)$ basic gates per column, for a total of
$\tilde{O}(dm)$ gates. This is intuitive as $A_B$ has $O(dm)$ nonzero
elements, so its gate complexity is what we would expect from a lookup
table (up to polylogarithmic factors).

To implement $P_b$ we need to construct the state vector $\ket{b}$. It
is known that this can be done using $\tilde{O}(m)$ basic gates
\citep{grover2002creating}; below, we formalize the fact that if $b$ is
$d$-sparse, $\tilde{O}(d)$ gates are sufficient.
\begin{proposition}
  \label{prop:sparsevec}
  Let $b \in \R^{m}$ be a vector with $d$ nonzero elements. Then
  the state $\ket{b}$ can be prepared with $\tilde{O}(d)$ gates.
\end{proposition}
Note that for an efficient implementation of an oracle that implements
$P_b$ that for a superposition of r.h.s.~vectors we can save the
intermediate data used in the construction described in the proof;
this takes space $\tilde{O}(d)$ for each vector, and in the rest of
the paper we frequently ignore similar details as long as they do not
affect the running time in $\tilde{O}$ notation.

\subsection{Sign estimation}
\label{sec:signest}
We first define a (well-known) subroutine to make two quantum states
interfere via Hadamard gates, creating the element-wise sum and
difference of the amplitudes; the subroutine is used several times
throughout the paper, so it is convenient to define a shorthand for
it.

\begin{algorithm}
  \caption{Interference routine: {\sc Interfere}$(U, V)$.}
  \label{alg:interfere}
    \begin{algorithmic}[1]
      \STATE {\bf Input:} two controlled state preparation unitaries
      $U, V$ on $q$ qubits with $U \ket{0^q} = \sum_{j = 0}^{2^q - 1}
      \alpha_j \ket{j}$ and $V \ket{0^q} = \sum_{j = 0}^{2^q - 1}
      \beta_j \ket{j}$.

      \STATE {\bf Output:} State $\ket{\psi} = \frac{1}{2} \ket{0} \otimes \sum_{j=0}^{2^q-1} (\alpha_j + \beta_j) \ket{j} + \frac{1}{2} \ket{1} \otimes \sum_{j=0}^{2^q-1} (\beta_j - \alpha_j) \ket{j}$.

      \STATE Introduce an auxiliary qubit in state $\ket{0}$ and apply
      a Hadamard gate, yielding the state $\frac{1}{\sqrt{2}}(\ket{0}
      + \ket{1})$; assume that this is the first qubit.

      \STATE Apply the controlled unitaries $\ket{0}\bra{0}
      \otimes I^{\otimes q} + \ket{1}\bra{1} \otimes U$ and
      $\ket{0}\bra{0}\otimes V + \ket{1}\bra{1} \otimes I^{\otimes q}$
      to the state, obtaining
      $\frac{1}{\sqrt{2}}(\ket{0}\otimes V\ket{0^q} + \ket{1} \otimes
      U\ket{0^q})$.

      \STATE Apply a Hadamard gate on the first qubit. Return the
      resulting state $\ket{\psi}$.
    \end{algorithmic}
\end{algorithm}
\begin{proposition}
  \label{prop:interfere}
  Let $U \ket{0^q} = \sum_{j = 0}^{2^q - 1} \alpha_j \ket{j}$ and $V
  \ket{0^q} = \sum_{j = 0}^{2^q - 1} \beta_j \ket{j}$. Then {\sc
    Interfere}$(U, V)$ (Alg.~\ref{alg:interfere}) returns the state
  $\ket{\psi} = \frac{1}{2} \ket{0} \otimes \sum_{j=0}^{2^q-1}
  (\alpha_j + \beta_j) \ket{j} + \frac{1}{2} \ket{1} \otimes
  \sum_{j=0}^{2^q-1} (\beta_j - \alpha_j) \ket{j}$, using one call to
  controlled-$U$ and controlled-$V$, and two extra gates.
\end{proposition}
The proof follows immediately by definition of the Hadamard gate $H
= \frac{1}{\sqrt{2}}\begin{pmatrix} 1 & 1 \\ 1 & -1 \end{pmatrix}$.

To determine the sign of a reduced cost we use the subroutines
{\sc SignEstNFN}, given in Alg.~\ref{alg:signestnfn}, and {\sc
  SignEstNFP}, given in Alg.~\ref{alg:signestnfp}. The acronyms
``NFN'' an ``NFP'' stand for ``no false negatives'' and ``no false
positives'', respectively, based on an interpretation of these
subroutines as classifiers that need to assign a 0-1 label to the
input. We need two such subroutines because the quantum phase
estimation, on which they are based, is a continuous
transformation. Therefore, a routine that has ``no false negatives'',
i.e., with high probability it returns 1 if the input data's true
class is 1 (in our case, this means that a given amplitude is $\ge
-\epsilon$), may have false positives: it may also return 1 with too
large probability for some input that belongs to class 0 (i.e., the
given amplitude is $< -\epsilon$). The probability of these
undesirable events decreases as we get away from the threshold
$-\epsilon$. We therefore adjust the precision and thresholds used in
the sign estimation routines to construct one version that with high
probability returns 1 if the true class is 1 (no false negatives), and
one version that with high probability returns 0 if the true class is
0 (no false positives). The subroutines are similar to the traditional
Hadamard test, but rather than taking a measurement, we apply
amplitude estimation on the Hadamard qubit.

\begin{algorithm}
  \caption{Sign estimation routine: {\sc SignEstNFN}$(U, k, \epsilon)$.}
  \label{alg:signestnfn}
    \begin{algorithmic}[1]
      \STATE {\bf Input:} state preparation unitary $U$ on $q$ qubits
      (and its controlled version) with $U \ket{0^q} = \sum_{j =
        0}^{2^q - 1} \alpha_j \ket{j}$ and $\alpha_j$ real, index $k
      \in \{0,\dots,2^q-1\}$, precision $\epsilon$.

      \STATE {\bf Output:} 1 if $\alpha_k \ge -\epsilon$, with
      probability at least $3/4$.

      \STATE \label{al:contrup} Let $V$ map $\ket{0^q}$ to
      $\ket{k}$. Compute $\ket{\psi} = \textsc{Interfere}(U, V)$.

      \STATE Apply amplitude estimation to the state $\ket{0} \otimes
      \ket{k}$ with $\ceil{\log \frac{\sqrt{3}\pi}{\epsilon}} + 2$ qubits of
      accuracy; let $\ket{a}$ be the bitstring produced by the phase
      estimation portion of amplitude estimation.

      \STATE If $0.a < \frac{1}{2}$, return 1 if $0.a \ge \frac{1}{6}
      - \frac{2\epsilon}{\sqrt{3}\pi}$, 0 otherwise; if $0.a \ge
      \frac{1}{2}$, return 1 if $1 - 0.a \ge \frac{1}{6} -
      \frac{2\epsilon}{\sqrt{3}\pi} $, 0 otherwise.
      
    \end{algorithmic}
\end{algorithm}
\begin{proposition}
  \label{prop:signestnfn}
  Let $U \ket{0} = \ket{\psi} = \sum_{j = 0}^{2^q - 1} \alpha_j
  \ket{j}$ with real $\alpha_j$, and $k
  \in \{0,\dots,2^q-1\}$ a basis state index. Then with probability at
  least 3/4:
  \begin{itemize}
  \item if $\alpha_k \ge -\epsilon$ Algorithm \ref{alg:signestnfn}
    ({\sc SignEstNFN}) returns $1$, and if Algorithm
    \ref{alg:signestnfn} returns $1$ then $\alpha_k \ge -2\epsilon$;
  \item if $\alpha_k < -2\epsilon$ Algorithm \ref{alg:signestnfn}
    ({\sc SignEstNFN}) returns $0$, and if Algorithm \ref{alg:signestnfn} returns $0$ then $\alpha_k < -\epsilon$.
  \end{itemize}
  The algorithm makes $O(1/\epsilon)$ calls to $U,U^{\dag}$ or their
  controlled version, and requires $O(q + \log^2 \frac{1}{\epsilon})$
  additional gates.
\end{proposition}

The alternative version of this routine, i.e., the one that with high
probability has ``no false positives'', can be constructed similarly,
by adjusting the thresholds. We give its pseudocode in the appendix.
\begin{proposition}
  \label{prop:signestnfp}
  Let $U \ket{0} = \ket{\psi} = \sum_{j = 0}^{2^q - 1} \alpha_j
  \ket{j}$ with real $\alpha_j$, and $k
  \in \{0,\dots,2^q-1\}$ a basis state index. Then with probability at
  least 3/4:
  \begin{itemize}
  \item if $\alpha_k \le -\epsilon$ Algorithm \ref{alg:signestnfp}
    ({\sc SignEstNFP}) returns $0$, and if Algorithm
    \ref{alg:signestnfp} returns $0$ then $\alpha_k \le \frac{1}{3}\epsilon$;
  \item if $\alpha_k > \frac{1}{3}\epsilon$ Algorithm \ref{alg:signestnfp}
    ({\sc SignEstNFP}) returns $1$, and if Algorithm
    \ref{alg:signestnfp} returns $1$ then $\alpha_k > -\epsilon$.
  \end{itemize}
  The algorithm makes $O(1/\epsilon)$ calls to $U,U^{\dag}$ or their
  controlled version, and requires $O(q + \log^2 \frac{1}{\epsilon})$
  additional gates.
\end{proposition}
It is important to discuss the assumption that $\alpha_j$ is real. In
this paper, the subroutines {\sc SignEstNFN} and {\sc SignEstNFP} are
applied to the output of a QLSA; by definition of the QLSA problem,
the output state is the amplitude encoding of the solution, see
\cite{childs2017quantum}, thus it has real coefficients for real
data. However, the implementation of a quantum algorithm as a circuit
is generally assumed to be modulo a global phase. Unfortunately, the
two subroutines described in this section (as well as the Hadamard
test) are sensitive to the relative phase of controlled-$U$ with
respect to the remaining part of the quantum state, so a global phase
in the implementation of $U$ would affect the result of the sign
test. This issue can be avoided by ensuring that the mapping
controlled-$V$ used in the algorithm has the same phase as
controlled-$U$, so that there is no relative phase. From a theoretical
point of view, a careful implementation of a QLSA (e.g., through the
Chebyshev polynomial \cite{childs2017quantum} or the singular value
transformation \cite{gilyen2019quantum} framework) is able to track
the phase applied by the algorithm, and therefore we can account for
this phase in subsequent steps: the only restriction is that we cannot
take the controlled-QLSA unitary to be a ``black-box'', and need to
know its phase.

\subsection{Determining if the basis is optimal, or the variable
  entering the basis}
\label{sec:pricing}
\begin{algorithm}
  \caption{Determining the reduced cost of a column:
    {\sc RedCost}$(A_B, A_k, c, \epsilon)$.}
  \label{alg:redcost}
    \begin{algorithmic}[1]
      \STATE {\bf Input:} basis $A_B$ with $\|A_B\| \le 1$,
      nonbasic column $A_k$, cost vector $c$ with $\|c_B\| = 1$,
      precision $\epsilon$.
      
      \STATE {\bf Output:} A quantum state such that the reduced cost
      of column $k$ is encoded in the amplitude of $\ket{0^{\ceil{\log
            m}}}$, and a flag indicating success, with bounded
      probability.

      \STATE \label{al:ls} Solve the linear system $\begin{pmatrix}
        A_B & 0 \\ 0 & 1 \end{pmatrix} \begin{pmatrix}x
        \\ y \end{pmatrix} = \begin{pmatrix} A_k \\ c_k \end{pmatrix}$
      using the quantum linear system algorithm with precision
      $\frac{\epsilon}{10\sqrt{2}}$. Let $\ket{\psi}$ be the
      quantum state obtained this way.

      \STATE \label{al:dotprod} Let $U_c$ be a unitary such that $U_c
      \ket{0^{\ceil{\log m + 1}}} = \ket{(-c_B,1)}$; apply $U_c^{\dag}$ to
      $\ket{\psi}$ to obtain $\ket{\psi'}$.

      \STATE Return $\ket{\psi'}$ and the flag indicating success as
      set by the quantum linear system algorithm.
    \end{algorithmic}
\end{algorithm}
\begin{algorithm}
  \caption{Determine if a column is eligible to enter the basis:
    {\sc CanEnter}$(A_B, A_k, c, \epsilon)$.}
  \label{alg:canenter}
    \begin{algorithmic}[1]
      \STATE {\bf Input:} Basis $A_B$ with $\|A_B\| \le 1$,
      nonbasic column $A_k$, cost vector $c$ with $\|c_B\| = 1$,
      precision $\epsilon$.
      
      \STATE {\bf Output:} 1 if the nonbasic column has reduced cost
      $< -\epsilon$, 0 otherwise, with bounded probability.

      \STATE Let $U_r$ be the unitary implementing {\sc RedCost}$(A_B,
      A_k, c, \epsilon)$.

      \STATE \label{al:signestcall} Return $1$ if {\sc SignEstNFN}$(U_r,
      0^{\ceil{\log m + 1}}, \frac{11 \epsilon}{10 \sqrt{2}} )$ is 0 and
      the success flag for {\sc RedCost}$(A_B, A_k, c, \epsilon)$ is
      1; return $0$ otherwise.

    \end{algorithmic}
\end{algorithm}

To find a column that can enter the basis, we use the subroutines {\sc
  RedCost} and {\sc CanEnter}, described in Alg.s~\ref{alg:redcost},
\ref{alg:canenter}. They are building blocks of {\sc
  FindColumn}, detailed in Alg.~\ref{alg:findcolumn}.
\begin{proposition}
  \label{prop:redcost}
  With bounded probability, Algorithm \ref{alg:redcost} ({\sc
    RedCost}) returns a quantum state such that the amplitude
  $\alpha_0$ of the basis state $\ket{0^{\ceil{\log m + 1}}}$
  satisfies:
  \begin{equation*}
    \left| \alpha_0 - \frac{\bar{c}_k}{\|(A_B^{-1}A_k, c_k)\|} \right| \le \frac{\epsilon}{10}.
  \end{equation*}
  The gate complexity of the algorithm is
  $O(\frac{1}{\epsilon}T_{\text{LS}}(A_B, A_k, \allowbreak
  \frac{\epsilon}{2}) + m)$.
\end{proposition}

\begin{proposition}
  \label{prop:canenter}
  With bounded probability, if Algorithm \ref{alg:canenter} ({\sc
    CanEnter}) returns 1 then column $A_k$ has reduced cost $<
  -\epsilon \|(A_B^{-1} A_k, c_k)\|$. The gate complexity of the
  algorithm is $O(\frac{1}{\epsilon}T_{\text{LS}}(A_B, A_k,
  \allowbreak \frac{\epsilon}{2}) + \log^2 \frac{1}{\epsilon} + m)$.
\end{proposition}
In Prop.~\ref{prop:canenter} we are concerned with ensuring that only
columns with negative reduced cost are returned, but using
Prop.~\ref{prop:signestnfn} it is easy to characterize an implication
in the opposite direction as well; in this specific case, it can be
easily established that if column $A_k$ has reduced cost $< -2.2
\epsilon \|(A_B^{-1} A_k, c_k)\|$, then {\sc CanEnter}$(A_B, A_k, c,
\epsilon)$ returns 1 with bounded probability.

\begin{algorithm}
  \caption{Determine the next column to enter the basis:
    {\sc FindColumn}$(A, B, c, \epsilon)$.}
  \label{alg:findcolumn}
    \begin{algorithmic}[1]
      \STATE {\bf Input:} Matrix $A$, basis $B$, cost vector $c$,
      precision $\epsilon$, with $\|A_B\| \le 1$ and $\|c_B\| = 1$.
      
      \STATE {\bf Output:} index of a column $k$ with $\bar{c}_k <
      -\epsilon \|(A_B^{-1} A_k, c_k)\|$ if one
      exists, with bounded probability.

      \STATE Let $U_{\text{rhs}}$ be the unitary that maps $\ket{k}
      \otimes \ket{0^{\ceil{\log m}}} \to \ket{k} \otimes \ket{A_k}$ for any
      $k \in \{1,\dots,n\} \setminus B$.

      \STATE \label{al:qsearch} Apply the quantum search algorithm
      (Thm.~\ref{thm:ampamp}) with search space $\{1,\dots,n\}
      \setminus B$ to the function {\sc CanEnter}$(A_B,
      U_{\text{rhs}}(\cdot), c, \epsilon)$.

      \STATE Return $k$ as determined by the quantum search algorithm.

    \end{algorithmic}
\end{algorithm}

\begin{theorem}
  \label{thm:findcolumn}
  Algorithm \ref{alg:findcolumn} ({\sc FindColumn}) returns a column
  $k \in N$ with reduced cost $\bar{c}_k < -\epsilon \|(A_B^{-1} A_k,
  \frac{c_k}{\|c_B\|})\|$, with expected number of iterations
  $O(\sqrt{n})$. The total gate complexity of the algorithm is
  $\tilde{O}(\frac{\kappa d \sqrt{n}}{\epsilon}(d_c n + d m))$.
\end{theorem}
We remark that Thm.~\ref{thm:findcolumn} concerns the case in which at
least one column eligible to enter the basis (i.e., with negative
reduced cost) exists. Following Alg.~\ref{alg:simplex_iter}, {\sc
  SimplexIter}, the subroutine {\sc FindColumn} is executed only if
{\sc IsOptimal} returns false.

The subroutine {\sc IsOptimal} can be constructed in almost the same
way as {\sc FindColumn}, hence we do not give the pseudocode. There
are two differences. First, at line \ref{al:signestcall} of {\sc
  CanEnter} we use {\sc SignEstNFP}: this ensures that if the
(rescaled) reduced cost is $\le -\epsilon$, {\sc CanEnter} returns
1 by Prop.~\ref{prop:signestnfp}. Second, at line \ref{al:qsearch} of
{\sc FindColumn}, rather than calling the quantum search algorithm, we
use amplidute estimation (Thm.~\ref{thm:ampest}) to determine if there is
any index $k$ for which (modified) {\sc CanEnter}$(A_B,
U_{\text{rhs}}(k), c, \epsilon)$ has value 1; if no such index exists,
{\sc IsOptimal} returns 1. If all algorithms are successful, which
happens with large probability, and {\sc IsOptimal} returns 1, we are
guaranteed that the current basis $B$ is optimal. The gate complexity
of {\sc IsOptimal} is exactly the same as in
Thm.~\ref{thm:findcolumn}, see Thm.~\ref{thm:ampamp}.

We can also examine the two possible cases of failure. A failure might
occur if the basis is optimal, but {\sc IsOptimal} returns 0. In this
case, failure can be detected and recovered from, because by
Prop.~\ref{prop:canenter}, there can be no index $k$ for which {\sc
  CanEnter} returns 1, so we observe 0 in the output register of
{\sc CanEnter} --- again, assuming all algorithms are successful.
Another failure might occur if the basis is not optimal, but {\sc
  FindColumn} fails to return a column with negative reduced
cost. This could happen if all columns have (rescaled) reduced cost
close to the threshold $-\epsilon$ (i.e., between $-2.2\epsilon$ and
$-\epsilon$). We can recover from this failure using {\sc SignEstNFP}
at line \ref{al:signestcall} of {\sc CanEnter}, ensuring that {\sc
  FindColumn} returns some index with {\sc CanEnter} equal to 1.

Note that when {\sc IsOptimal} returns 1, the current
basis $B$ is optimal but we do not have a classical description of the
solution vector. It is straightforward to write a quantum subroutine
to compute the optimal objective function value, using the techniques
discussed in this paper. To obtain a full description of the solution
vector, the fastest approach is to classically solve the system $A_B x
= b$, which requires time $O(d_c^{0.7} m^{1.9} + m^{2 + o(1)})$. We
remark that this is only necessary in the last iteration, hence it
does not dominate the running time unless the problem instance is
trivial.

The column selection subroutine {\sc FindColumn} presented above
returns a randomly chosen column with negative reduced cost. This is
sufficient for convergence. To implement a pivoting rule akin to
Dantzig's rule, i.e., that selects the column with minimum reduced
cost, a bit more work is needed. To use the quantum minimum finding
algorithm (Thm.~\ref{thm:minfind}) we need a subroutine that can
approximately compare the reduced costs of two columns, say $j$ and
$k$. For the normalized reduced cost discussed so far, i.e.,
$\bar{c}_j/\|(A_B^{-1}A_j, c_j)\|$, this is straightforward: we use
subroutine {\sc RedCost} to compute the normalized reduced cost of
$A_j$ and $A_k$, {\sc Interfere} to encode their difference, and {\sc
  SignEstNFN} to determine if the difference is at least
$\epsilon$. This gives the same running time as in
Thm.~\ref{thm:findcolumn}. If we want to implement Dantzig's rule for
the original reduced costs we need to remove the normalization before
applying {\sc Interfere}: this increases the cost by a factor
$\max_{j=1,\dots,n} c_j$. The implementation details are very similar
to those for the steepest edge pivoting rule, that we discuss next.

To implement the steepest edge pivoting rule
\citep{forrest1992steepest} we apply quantum minimum finding with the
comparison subroutine for the reduced cost described in
Alg.~\ref{alg:secompare}.
\begin{algorithm}
  \caption{Compare the steepest edge price of two columns:
    {\sc SteepestEdgeCompare}$(A_B, A_j, A_k, c, \epsilon)$.}
  \label{alg:secompare}
  \begin{algorithmic}[1]
      \STATE {\bf Input:} Basis $A_B$ with $\|A_B\| \le 1$, columns
      $A_j, A_k$, cost vector $c$ with $\|c_B\| = 1$, precision
      $\epsilon$.
      
      \STATE {\bf Output:} 1 if $\frac{\bar{c}_k}{A_B^{-1}A_k} \le
      (1-\epsilon)\frac{\bar{c}_j}{A_B^{-1}A_j} - \epsilon$, 0
      otherwise.

      \STATE \label{al:normest1} Compute an estimate of the norms
      $\|(A_B^{-1}A_j, c_j)\|$, $\|(A_B^{-1}A_k, c_k)\|$ in separate
      registers, with relative error $\frac{\epsilon}{4}$. Call them
      $\tilde{d}_j, \tilde{d}_k$ respectively.

      \STATE \label{al:normest2} Compute an estimate of the norms
      $\|A_B^{-1}A_j\|$, $\|A_B^{-1}A_k\|$ in separate
      registers, with relative error $\frac{\epsilon}{4}$. Call them
      $\tilde{e}_j, \tilde{e}_k$ respectively.
      
      \STATE \label{al:lsse1} Let $U$ be the unitary that applies {\sc
        RedCost}$(A_B, A_j, c, \epsilon/4)$, then multiplies the
      coefficient of $\ket{0^{\ceil{\log m}}}$ by
      $\frac{\tilde{d}_j}{\tilde{e}_j \max_{\ell \in N} c_\ell}$.

      \STATE \label{al:lsse2} Let $V$ be the unitary that applies {\sc
        RedCost}$(A_B, A_k, c, \epsilon/4)$, then multiplies the
      coefficient of $\ket{0^{\ceil{\log m}}}$ by
      $\frac{\tilde{d}_k}{\tilde{e}_k \max_{\ell \in N} c_\ell}$.

      \STATE \label{al:interferese} Let $W$ be the unitary that
      applies {\sc Interfere}$(U, V)$. 

      \STATE Return 1 if {\sc SignEstNFN}$(W, 0^{\ceil{\log m}+1},
      \epsilon/(8\max_{\ell \in N} c_\ell))$ is 0 and the success flag
      for all algorithms is 1.
    \end{algorithmic}
\end{algorithm}
\begin{theorem}
  \label{thm:findcolumnse}
  Let $h := \arg \min_{h \in N} \frac{c_h -
    c_B^{\top}A_B^{-1}A_h}{\|A_B^{-1}A_h\|}$. Using quantum minimum
  finding with the comparison subroutine {\sc SteepestEdgeCompare}$(A_B,
  A_j, A_k, c_B, \epsilon)$, optimizing over the set $N$, we can
  determine a column index with steepest edge price at most:
  \begin{equation*}
    (1+\epsilon)\frac{c_h -
      c_B^{\top}A_B^{-1}A_h}{\|A_B^{-1}A_h\|} + \epsilon,
  \end{equation*}
  with $O(\sqrt{n})$ iterations. The total gate complexity of the
  algorithm is $\tilde{O}(\frac{\kappa c_{\max} d
    \sqrt{n}}{\epsilon}\allowbreak (d_c n + d m))$, where $c_{\max} :=
  \max_{\ell \in N} c_\ell$.
\end{theorem}
We remark that the normalization factor $c_{\max}$, which appears in
the running time, is used to ensure that the (normalized) steepest
edge prices are at most 1 in absolute value, so that they can be
encoded in an amplitude. It is not difficult to change the algorithm
to determine if a potentially smaller normalization factor is
sufficient, but since in the worst case the normalization is of the
order of $c_{\max}$, we do not pursue this approach.

\subsection{Improving the running time when the ratio $n/m$ is large}
\label{sec:wide}
The pricing algorithms discussed in Sect.~\ref{sec:pricing} highlights
a tradeoff between the number of iterations and time necessary to load
the data. Indeed, if we apply the quantum search algorithm over all
columns we need to perform $O(\sqrt{n})$ iterations, but the unitary
to prepare the data for the QLSA requires time that scales as
$\tilde{O}(d_c n)$. In some cases it may therefore be avantageous to
split the set of columns into multiple sets, and apply the search
algorithm to each set individually. To formalize this idea, suppose
that we split the $n$ columns into $h$ sets of equal size. The running
time of the algorithm discussed in Thm.~\ref{thm:findcolumn} then
becomes $\tilde{O}(\frac{h d \kappa}{\epsilon} \sqrt{\frac{n}{h}} (d_c
\frac{n}{h} + d m))$. We can determine the optimal $h$ by
minimizing the above expression. Ignoring the multiplication factor
and setting the derivative with respect to $h$ equal to zero, the
optimal $h$ must satisfy:
\begin{align*}
  -\frac{1}{2} h^{-3/2} n^{3/2} d_c + \frac{1}{2} h^{-1/2} d m \sqrt{n} = 0 ,
\end{align*}
which yields $h = \frac{d_c n}{d m}$. Since $h$ must be integer, we
can use this approach only if $\frac{n}{m} \ge 2 \frac{d}{d_c}$;
hence, $n/m$ may have to be large if the columns can be much sparser
than the rows. Substituting the optimal $h$ in the running time
expression computed above yields
$\tilde{O}(\frac{1}{\epsilon}\kappa d^{1.5} \sqrt{d_c} n
\sqrt{m})$. Notice that this choice of $h$ is optimal for
Thm.~\ref{thm:findcolumnse} as well, as the correponding running time
is simply multiplied by $c_{\max}$, yielding a total running time of
$\tilde{O}(\frac{1}{\epsilon}c_{\max}\kappa d^{1.5} \sqrt{d_c} n \sqrt{m})$
with steepest edge pivoting.

\subsection{Estimating the error tolerance}
\label{sec:normest}
As detailed in Sect.~\ref{sec:pricing}, the quantum pricing algorithm
with the equivalent of a Dantzig rule uses the optimality tolerance
$-\epsilon \|(A_B^{-1} A_k, \allowbreak \frac{c_k}{\|c_B\|})\|$. It
may therefore be desirable to compute some estimate of $\|A_B^{-1}
A_k\|$, since the optimality tolerance used by the algorithm depends
on it. However, $A_B^{-1}$ is not classically available. We give a
quantum subroutine to compute the root mean square of $\|A_B^{-1}
A_k\|$ over all $k$.
\begin{proposition}
  \label{prop:normest}
  Let $\|A_B\| = 1$. Using amplitude estimation, it is possible to
  estimate $\|A_B^{-1} A_N\|^2_F$ up to relative error $\epsilon$ with
  gate complexity $\tilde{O}(\frac{\kappa^2}{\epsilon}(\kappa n d_c +
  \kappa^2 d^2 m))$.
\end{proposition}
Given an estimate of $\|A_B^{-1} A_N\|^2_F$ with relative error
$\epsilon$, dividing by $|N|$ and taking the square root yields an
estimate of the root mean square of $\|A_B^{-1} A_k\|$ for $k \in N$,
with relative error $\approx \sqrt{\epsilon}$ (up to first order
approximation). Obviously, we could also use Thm.~\ref{thm:qlsnormest}
to estimate just the norm of $\|A_B^{-1} A_k\|$ for a specific column
$k$, e.g., the column entering the basis. 

%% Then, we discuss the computation of $\max_{k \in N}
%% \|(A_B^{-1} A_k,c_k)\|$; computing the minimum can be done in a
%% similar way. As detailed in the proof of Prop.~\ref{prop:normest}, we can use amplitude estimation together with the QLSA to compute

\subsection{Determining the variable leaving the basis}
\label{sec:ratio_test}
In this section we use subroutines {\sc SignEstNFN}${}^+(U, k, \epsilon)$
and {\sc SignEstNFP}${}^+(U, k, \epsilon)$ that check if a real
amplitude (up to global phase) has positive sign, i.e., they return
$1$ if $\alpha_k \ge \epsilon$, $0$ otherwise, with the same structure
as the subroutines described in Sect.~\ref{sec:signest} (i.e., {\sc NFN}
does not have false negatives, {\sc NFP} does not have false
positives). Such subroutines can easily be constructed in a manner
similar to {\sc SignEstNFN}$(U, k, \epsilon)$ and {\sc SignEstNFP}$(U,
k, \epsilon)$: rather than estimating the amplitude of the basis state
$\ket{0} \otimes \ket{k}$ in the proof of Prop.~\ref{prop:signestnfn}
and Prop.~\ref{prop:signestnfp}, we estimate the amplitude of
$\ket{1} \otimes \ket{k}$ and adjust the return value of the algorithm
(e.g., for {\sc SignEstNFN}${}^+(U, k, \epsilon)$, if $0.a <
\frac{1}{2}$, we return 1 if $0.a \le \frac{1}{6} - 2
\frac{\epsilon}{\sqrt{3}\pi}$, 0 otherwise). The gate complexity is
the same.

\begin{algorithm}
  \caption{Determine if the problem is unbounded from below:
    {\sc IsUnbounded}$(A_B, A_k, \delta)$.}
  \label{alg:isunbounded}
    \begin{algorithmic}[1]
      \STATE {\bf Input:} Basis $A_B$ with $\|A_B\| \le 1$, nonbasic
      column $A_k$, precision $\delta$.
      
      \STATE {\bf Output:} 1 if $A_B^{-1} A_k < \delta \mathbf{1}_m
      \|A_B^{-1} A_k\|$, 0 otherwise, with bounded probability.

      \STATE Let $U_{\text{LS}}$ be the unitary implementing the QLSA
      for the system $A_B x = A_k$ with precision $\frac{\delta}{10}$.

      \STATE Define a function $g(\ell) :=$ {\sc
        SignEstNFN}${}^+(U_{\text{LS}}, \ell, \frac{9\delta}{10}) \land $
      (the success flag for QLSA is 1).
      
      \STATE Use amplitude estimation (Thm.~\ref{thm:ampest}) to
      determine if there exists $\ell : 1,\dots,m$ such that $g(\ell)
      = 1$.

      \STATE If no such $\ell$ is found, return 1; otherwise, return
      0.
    \end{algorithmic}
\end{algorithm}
We first describe an algorithm to determine if column $k$ of the LP
proves unboundedness, then describe how to perform the ratio test. 
\begin{proposition}
  \label{prop:isunbounded}
  With bounded probability, if Algorithm \ref{alg:isunbounded} ({\sc
    IsUnbounded}) returns 1 then $A_B^{-1} A_k < \delta \mathbf{1}_m
  \|A_B^{-1} A_k\|$, with total gate complexity
  $\tilde{O}(\frac{1}{\delta}\kappa d^2 m^{1.5})$. Choosing $\delta =
  1/\|A_B^{-1}A_k\|$, we can test if $A_B^{-1} A_k < \delta
  \mathbf{1}_m$ with total gate complexity
  $\tilde{O}(\frac{\eta}{\delta}\kappa^2 d^2 m^{1.5})$.
\end{proposition}
If {\sc IsUnbounded}$(A_B, A_k, \delta)$ returns 1, we have a proof
that the LP is unbounded from below, up to the given
tolerance. Otherwise, we have to perform the ratio test. This is
described in the next subroutine. On line \ref{al:binsearch} of the
subroutine, we omit several details for readability; in particular, we
use slightly different unitaries $U_r$ to test whether we are above or
below the threshold $-\delta/2$. The details are given in the proof of
Thm.~\ref{thm:qratio_test}.
\begin{algorithm}
  \caption{Determine the basic variable (row) leaving the basis:
    {\sc FindRow}$(A_B, A_k, b, \delta, t)$.}
  \label{alg:findrow}
    \begin{algorithmic}[1]
      \STATE {\bf Input:} Basis $A_B$ with $\|A_B\| \le 1$, nonbasic
      column $A_k$, r.h.s.\ $b$, precision $\delta$.
      
      \STATE {\bf Output:} index of the row that should leave the
      basis according to the ratio test \eqref{eq:ratio_test}, with
      bounded probability.

      \STATE \label{al:compfind} For $r \ge 0$, define a unitary $U_r$
      that performs the following operations: it applies the QLSA to
      solve the system $A_B x = b - rA_k$, then uses amplitude
      estimation (Thm.~\ref{thm:ampest}) and {\sc SignEstNFN} with
      precision $O(\frac{\delta}{\|A_B^{-1}A_k\|})$ to determine if
      any component $\ell=1,\dots,m$ of the solution vector are $<
      -\frac{\delta}{2}$. If so it returns 1 and a corresponding
      index $j$, otherwise it returns 0.

      \STATE \label{al:binsearch} Perform binary search on $r \ge 0$
      to determine a value $\tilde{r}$ such that $U_{\tilde{r}}$
      returns 1 and $U_{\tilde{r}-\frac{\delta}{2\kappa\|A_k\|}}$
      returns 0. 

      \STATE Return the row index $\ell$ identified by $U_{\tilde{r}}$.
    \end{algorithmic}
\end{algorithm}

\begin{theorem}
  \label{thm:qratio_test}
  Let $\ell$ be the row index returned by Algorithm \ref{alg:findrow}
  ({\sc FindRow}), and let $\hat{B}$ be the basis obtained from $B$ by
  replacing column $k$ with the basic variable corresponding to row
  $\ell$. Then, with bounded probability, the basic solution
  corresponding to $\hat{B}$ is $\delta$-feasible, i.e.,
  $A_{\hat{B}}^{-1} b \ge -\delta \mathbf{1}_m$.  The gate complexity
  of the algorithm is $\tilde{O}(\frac{1}{\epsilon} \eta d^2 \kappa^2 m^{1.5})$.
\end{theorem}
Notice that the ratio test is performed approximately, i.e., the
solution found after pivoting might be infeasible, but the maximum
infeasibility after pivoting is bounded by $\delta$.  We discuss how
to deal with slightly infeasible solutions in Sect.~\ref{sec:harris}.

%% There is only one situation in which the algorithms described in this
%% section may fail: if {\sc IsUnbounded} returns 0, but on line
%% \ref{al:ratiocomp} of {\sc FindRow} the condition
%% $\textsc{SignEstNFP}^+(U_{\text{LS}}, h, \frac{\delta}{2}) = 1$ is
%% never satisfied. In this case {\sc FindRow} can return a failure flag,
%% and we can decide how to proceed; one possible way would be to relax
%% the sign check condition slightly, and another way would be to ignore
%% the failure since it is an indicator of numerical issues (all
%% nonnegative values at the denominator of the ratio test are very close
%% to 0: even if the LP may not be unbounded, it is close to being so,
%% and ill-conditioned).

\subsection{Two-pass ratio test: improving condition number, and dealing with infeasibilities}
\label{sec:harris}
The classical simplex method benefits from decades of experience,
yielding multiple modifications that, while not necessary in theory,
are fundamental for its practical success. One such modification is
the two-pass Harris ratio test, a description of which can be found in
\citep{gill1989practical}. This modification to the pivoting step is
known to improve the condition number of the bases explored in the
course of the algorithm. The two-pass Harris ratio test works as
follows (to be consistent with the rest of this paper we describe it
for a problem in standard form, but it can easily be extended to
problems with lower and upper bounds on the decision variables). Let
$\delta$ be the feasibility tolerance. Instead of using
\eqref{eq:ratio_test}, we compute, in the first pass:
\begin{equation}
  \label{eq:harris_ratio_test}
  \bar{r} := \min_{j=1,\dots,m : u_j > 0} \frac{x_{B(j)} + \delta}{u_j}.
\end{equation}
This is equivalent to relaxing $x \ge 0$ to $x \ge -\delta \mathbf{1}_n$. Then, in the second pass, we determine:
\begin{equation}
  \label{eq:harris_ratio_test_2}
  \ell := \arg\max_{j=1,\dots,m : u_j > 0} \{u_j :
  \frac{x_{B(j)}}{u_j} \le \bar{r}\}.
\end{equation}
In other words, we first determine the minimum value of a relaxed
ratio test, then we choose the pivot row as the one that maximizes the
value of the pivot element $u_j$, while still giving a ratio that does
not exceed the relaxed ratio test value \eqref{eq:harris_ratio_test},
thereby ensuring that $\delta$-feasibility is preserved.

Implementing this strategy as part of our quantum subroutines is
straightforward. Algorithm~\ref{alg:findrow} already uses the relaxed
feasibility definition $x \ge -\delta \mathbf{1}_n$, see
Thm.~\ref{thm:qratio_test}. Let $\tilde{r}$ be the value computed on
line \ref{al:binsearch} of Algorithm~\ref{alg:findrow}. We can use
repeated applications of $U_{\tilde{r}}$ to return all row indices
$\ell$ that satisfy $(A_B^{-1}(b - \tilde{r} A_k))_{\ell} <
-\frac{\delta}{10}$ or a similar tolerance. If there are $t$ such
indices, this increases the running time by a factor $\sqrt{t}$, see,
e.g., \cite{ambainis2004quantum}; in practice we expect $t$ to be
small, and we do not have to find all indices if it becomes too time
consuming. Let $L$ be the set of such indices. We can find the maximum
pivot element among indices in $L$ by using approximate quantum
minimum finding (Thm.~\ref{thm:minfind}) combined with the sign
estimation subroutines of Sect.~\ref{sec:signest}, in a manner similar
to the construction discussed in Thm.~\ref{thm:findcolumnse}. The
running time of this last step is essentially the same as in
Thm.~\ref{thm:qratio_test}.

Similarly to practical implementations of the classical simplex
algorithm, the basis obtained after a pivot may be slightly
infeasible, as controlled by the feasibility tolerance $\delta$. It is
therefore important to discuss how to recover from this situation. We
propose two approaches. The first approach is the implementation of a
strict zero step for basic variables that are negative after a pivot;
see \cite{gill1989practical,huangfu2013high}. To do so, we simply
estimate the value of the basic variable that is leaving the basis,
i.e., $(A_B^{-1}b)_\ell$ where $\ell$ is the index returned by
Algorithm~\ref{alg:findrow}. Such an estimate can be computed with
Euclidean norm error $\delta$ in time $\tilde{O}(\frac{\eta d^2
  \kappa^2}{\delta}m)$, which is faster than
Algorithm~\ref{alg:findrow}. If the variable leaving the basis has a
negative value, we fix it to its slightly infeasible value until it
enters the basis again, or the basis is recomputed (see below). Note
that fixing a nonbasic variable to some nonzero value is equivalent to
adjusting the right-hand side $b$.

The second approach is a periodic recomputation of the basis to
eliminate any infeasibilities. More specifically, given a basis $B$,
we determine if it is feasible according to a specified tolerance, and
if it is not, we switch to Phase 1 of the simplex method, possibly
increasing precision, to attain feasibility. Since we can increase the
precision arbitrarily, we can assume that Phase 1 is always successful
if the original problem admits a feasible basis. Note, however, that
if Phase 1 is not able to regain feasibility within a number of
iterations smaller than the basis recomputation frequency, the
algorithm fails (this can also happen in the classical simplex method
\citep{gill1989practical}). With the proposed quantum subroutines we
do not have classical, unrestricted access to the basic feasible
solution $A_B^{-1} b$. Thus, to check feasibility of the solution
associated with a basis $B$, we can use Algorithm~\ref{alg:isfeasible}
instead, keeping the same running time guarantees.
\begin{algorithm}
  \caption{Determine if a basic solution is feasible:
    {\sc IsFeasible}$(A_B, b, \delta)$.}
  \label{alg:isfeasible}
    \begin{algorithmic}[1]
      \STATE {\bf Input:} Basis $A_B$ with $\|A_B\| \le 1$, right-hand
      side vector $b$, precision $\delta$.
      
      \STATE {\bf Output:} 1 if $A_B^{-1} b \ge -\delta \mathbf{1}_m$,
      0 otherwise, with bounded probability.

      \STATE Let $U_{\text{LS}}$ be the unitary implementing the QLSA
      for the system $A_B x = b$ with precision $\frac{\delta}{10\|A_B^{-1}
      b\|}$.

      \STATE Define a function $g(\ell) := (\neg
      \textsc{SignEstNFP}(U_{\text{LS}}, \ell,
      \frac{9\delta}{20\|A_B^{-1} b\|})) \land $ (the success flag for
      QLSA is 1).
      
      \STATE Use amplitude estimation (Thm.~\ref{thm:ampest}) to
      determine if there exists $\ell \in \{1,\dots,m\}$ such that
      $g(\ell) = 1$.

      \STATE If no such $\ell$ is found, return 1; otherwise, return
      0.
    \end{algorithmic}
\end{algorithm}
\begin{proposition} 
  \label{prop:isfeasible}
  With bounded probability, if $A_B^{-1} b \not\ge -\delta
  \mathbf{1}_m$ then Algorithm \ref{alg:isfeasible} ({\sc IsFeasible})
  returns 0. The total gate complexity of the algorithm is
  $\tilde{O}(\frac{\eta}{\delta}\kappa^2 d^2 m^{1.5})$.
\end{proposition}
If {\sc IsFeasible} returns 1 and the algorithm is successful, the
basic solution is feasible with tolerance $\delta$, so that we do not
need to change the basis; otherwise, as discussed, we can switch to
Phase 1 aiming to minimize infeasibilities.

\section{Faster implementation with QRAM}
\label{sec:qram}
We now describe how to modify the quantum subroutines when QRAM is
available. To solve linear systems, we rely on a QLSA that constructs
a {\em block encoding} of the matrix from QRAM data structures, see
e.g., \citep{chakraborty2018power,gilyen2019quantum} for a
definition. We require a QRAM of size $\tilde{O}(d_cn + m)$ to store
$b$ as well as all the nonzero entries of $A$, arranged in a data
structure similar to the one described in
Prop.~\ref{prop:sparsevec}. For $p \in [0,1]$, define $\mu(A) :=
\min\left\{\|A\|_F, \sqrt{s_{2p}(A) s_{2(1-p)}(A^{\top})}\right\}$
where $S_{p}(A) = \max_i \sum_j A_{ij}^p$; the running time of the
QLSA depends on $\mu(A)$.
\begin{proposition}
  \label{prop:qlsaqram}
  If the matrix $A_B$ and the columns of $A_N$ are stored in QRAM, we
  can implement a QLSA with running time $T_{\text{LS}}(A_B, A_N,
  \epsilon) = \tilde{O}(\mu(A_B) \kappa)$. The cost of preparing the
  data structures before the first iteration of
  Alg.~\ref{alg:simplex_iter} is $\tilde{O}(d_cn)$; the time to update
  the data structures after the basis changes is $\tilde{O}(m)$.
\end{proposition}

An iteration of the simplex method proceeds exactly as described in
Alg.~\ref{alg:simplex_iter}, replacing each call to the QLSA of
\cite{childs2017quantum} in the subroutines with the QLSA of
\cite{chakraborty2018power}, as in Prop.~\ref{prop:qlsaqram}.

The running time of {\sc FindColumn} is $\tilde{O}(\frac{\sqrt{n}
}{\epsilon}(T_{\text{LS}}(A_B, A_k, \frac{\epsilon}{2}) + m))$, and by
Prop.~\ref{prop:qlsaqram} this is $\tilde{O}(\frac{\sqrt{n}
}{\epsilon}(\mu(A_B) \kappa + m))$. Notice that because $A_B$ is $m
\times m$ and $\|A_B\| = 1$, $\mu(A_B) \le \sqrt{m}$, hence we obtain
running time $\tilde{O}(\frac{1}{\epsilon}\kappa \sqrt{mn})$. The
running time calculation to determine if the basis is optimal is
similar. With the steepest edge version of pricing, the running time
increases by a factor $c_{\max}$ once again. The running time of {\sc
  FindRow} is $\tilde{O}(\frac{\eta\kappa}{\delta}\sqrt{m}
T_{\text{LS}}(A_B, A_k, \frac{\delta}{8\|A_B^{-1}(b-rA_k)\|}))$, and
by Prop.~\ref{prop:qlsaqram} this is
$\tilde{O}(\frac{\eta}{\delta}\sqrt{m} \kappa^2 \mu(A_B))$; using the
upper bound $\mu(A_B) \le \sqrt{m}$, we can upper bound it by
$\tilde{O}(\frac{\eta}{\delta} \kappa^2 m)$. Similar calculations for
the subroutines {\sc IsUnbounded} and {\sc IsFeasible} yield running
time $\tilde{O}(\frac{\eta}{\delta} \kappa^2 m)$. Finally, we remark
that at every iteration of Alg.~\ref{alg:simplex_iter}, on top of the
$\tilde{O}(md)$ classical operations to normalize data, we need to
update $\tilde{O}(m)$ memory locations in QRAM, as discussed in
Prop.~\ref{prop:qlsaqram}.

\section{Future research}
\label{sec:conclusions}
This paper proposes several quantum subroutines for simplex
method. For several reasons, it is conceivable that the factors
$\sqrt{n}, d_c n$ and $\kappa$ in the running time of the pricing step
cannot be further decreased when using a similar scheme to what is
presented in this paper. This is based on the following observations:
first, the dependence of QLSA cannot be improved to
$\kappa^{1-\delta}$ for $\delta > 0$ unless BQP = PSPACE
\citep{harrow2009quantum}; second, the factor $O(\sqrt{n})$ is optimal
for quantum search, relative to an oracle that identifies the
acceptable solutions \citep{bennett1997strengths}; and third, to be
able to determine if a column has negative reduced cost (in
superposition), we would expect that the running time cannot be
reduced below $O(d_c n)$ in general (since this is the number of
nonzero elements in the matrix). However, it is possible to exploit
structure in the constraint matrix to reduce this factor: this
requires specialized data preparation algorithms, but it may be worth
the effort for certain structures that appear frequently in
LPs. Regarding the condition number, even though improving the
theoretical dependency may not be possible, there could exist
efficient practical strategies to keep it under control: this is
routinely done in classical implementations of the simplex method, and
some examples are of this discussed in this paper. A systematic study
of such strategies, in theory and in practice, could prove
informative. Of course, it may be possible to give altogether better
algorithms using a different scheme than the one presented here.

%% There may also exist opportunities to speedup the ratio test. One
%% possibility could be to use a QLSA to compute $A_B^{-1} b - t A_B^{-1}
%% A_k$ for a fixed value of $t$, and then perform binary search on $t$
%% to determine for which nonnegative value of $t$ a component of the
%% vector becomes negative. Unfortunately, the machinery constructed in
%% Sect.~\ref{sec:qsimplex_details} is not efficient in this context, and
%% repeating the routine {\sc SignEst} for each component of $A_B^{-1} b
%% - t A_B^{-1} A_k$ would be slower than the algorithm for the ratio
%% test that we describe. Exploring more efficient algorithms for this
%% task this is left for future research.

%% Finally, improving efficiency when running several iterations of the
%% simplex method, as opposed to a single iteration, could provide
%% significant benefits as well as additional ways of exploiting
%% superposition, for example by ``looking ahead'' at the outcome of
%% several pivots so as to choose the most promising column with negative
%% reduced cost. While this does not lead to an efficient algorithm in
%% theory, for a constant number of pivots it could be numerically interesting.

\subsection*{Acknowledgments}
  We are grateful to Sergey Bravyi, Sanjeeb Dash,
  Santanu Dey, Yuri Faenza, Krzysztof Onak, Ted Yoder, and to
  anonymous referees for useful discussions and/or comments on an
  early version of this manuscript. The author is partially supported
  by the IBM Research Frontiers Institute, Army Research Office grant
  W911NF-20-1-0014, and AFRL grant FA8750-C-18-0098.

\bibliographystyle{apalike}
\bibliography{quantum,phd}

\appendix

\section{Proofs}

\subsection{Proofs from Section \ref{sec:preliminaries}.}
\begin{proof}{\it Proof of Thm.~\ref{thm:minfind}.}
  Note that, by virtue of how the comparison condition is structured,
  in each application of the quantum search algorithm there is always
  at least one ``good'' element. Thus, the proof and the running time
  analysis given by
  \cite{durr1996quantum} for the case $\epsilon = 0$ (i.e., exact
  minimum finding) also work for $\epsilon > 0$, using the same
  argument given by \cite[Sect.~3]{durr1996quantum} for the case in
  which the values of $f(x)$ are not distinct.
\end{proof}

\subsection{Proofs from Section \ref{sec:oracles}.}
\begin{proof}{\it Proof of Prop.~\ref{prop:sparsevec}.}
  We can rely on the construction of \cite{grover2002creating}, that
  consists in a sequence of controlled rotations. The construction can
  be understood on a binary tree with $m$ leaves, representing the
  elements of $b$, and where each inner node has value equal to the
  sum of the squares of the children nodes. Each inner node requires a
  controlled rotation with an angle determined by the square root of
  the ratio between the two children nodes. Assuming that $m$ is a
  power of 2 for simplicity, the tree has $\log m$ levels and the
  construction when $b$ is dense ($d = m$) requires $\sum_{j=0}^{\log
    m} 2^j = O(m)$ controlled rotations. Notice that at each inner
  node, a controlled rotation is necessary only if both children have
  nonzero value. If only one child has nonzero value the operation
  requires at most a controlled $X$, and if both children have value zero no
  operation takes place. If $d < m$, there can be at most $d$ nodes at
  each level that require a controlled rotation, and in fact the
  deepest level of the binary tree such that all nodes contain nonzero
  value is $\floor{\log d}$. We need $O(d)$ controlled rotations up to
  this level of the tree, and for each of these nodes we may need at
  most $O(\log m)$ susequent operations, yielding the total gate
  complexity $\tilde{O}(d)$. 
  \end{proof}

\subsection{Proofs from Section~\ref{sec:signest}.} 

\begin{proof}{\it Proof of Prop.~\ref{prop:signestnfn}.}

  After line \ref{al:contrup}, using Prop.~\ref{prop:interfere} we are
  in the state:
  \begin{align*}
    \frac{1}{2} (1 + \alpha_k) \ket{0} \otimes \ket{k} + \frac{1}{2} (1 - \alpha_k) \ket{1} \otimes \ket{k} + \frac{1}{2} (\ket{0} - \ket{1}) \otimes 
    \sum_{\substack{j=0 \\ j \neq k}}^{2^n-1} \alpha_j \ket{j}.
  \end{align*}
  We now apply amplitude estimation to the state
  $\ket{0}\otimes\ket{k}$ to determine the magnitude of $\frac{1}{2}
  (1 + \alpha_k)$. The exact phase that must be estimated by the
  phase estimation portion of the algorithm is the number $\theta$
  such that:
  \begin{equation*}
    \sin \pi \theta = \frac{1}{2} (1 + \alpha_k).
  \end{equation*}
  Suppose $\alpha_k \ge -\epsilon$. Then $\sin \pi \theta \ge
  \frac{1}{2} (1 - \epsilon)$, implying, by monotonicity of $\sin^{-1}$
  over its domain:
  \begin{align*}
    \theta > \frac{\sin^{-1} \left(\frac{1}{2} (1 - \epsilon)\right)}{\pi} \ge
    \frac{\frac{\pi}{6} + \frac{2}{\sqrt{3}}\left(\frac{1}{2} (1 - \epsilon) - \frac{1}{2} \right)}{\pi} \ge \frac{\frac{\pi}{6} - \frac{\epsilon}{\sqrt{3}}}{\pi}
    \ge \frac{1}{6} - \frac{\epsilon}{\sqrt{3}\pi},
  \end{align*}
  where for the second inequality we have used the Taylor expansion of $\sin^{-1}(x)$ at $x = \frac{1}{2}$:
  \begin{equation*}
    \sin^{-1}(x) \approx \frac{\pi}{6} + \frac{2}{\sqrt{3}}(x - \frac{1}{2}) + \frac{2\sqrt{3}}{9}(x - \frac{1}{2})^2.
  \end{equation*}
  Rather than $\theta$, we obtain an approximation $\tilde{\theta}$ up
  to a certain precision. By Thm.~\ref{thm:qpe}, using $\ceil{\log
    \frac{\sqrt{3}\pi}{\epsilon}} + 2$ qubits of precision, then
  \begin{equation}
    \label{eq:tildebound}
    |\theta - \tilde{\theta}| \le \frac{\epsilon}{\sqrt{3}\pi}
  \end{equation}
  with probability at least $3/4$. Then we must have $\tilde{\theta}
  \ge \frac{1}{6} - \frac{2\epsilon}{\sqrt{3}\pi}$ with probability at
  least 3/4. In this case, the algorithm returns 1 (recall that if the
  first bit of the expansion is 1, i.e., $0.a > \frac{1}{2}$, we must
  take the complement $1 - 0.a$ because we do not know the sign of
  eigenvalue on which phase estimation is applied, see
  \citep{brassard2002quantum} for details).

  Now suppose the algorithm returns 1, implying $\tilde{\theta} \ge
  \frac{1}{6} - \frac{2\epsilon}{\sqrt{3}\pi}$. By
  \eqref{eq:tildebound} we must have $\theta \ge \tilde{\theta} -
  \frac{\epsilon}{\sqrt{3}\pi} \ge \frac{1}{6} -
  \frac{3\epsilon}{\sqrt{3}\pi}$. Thus,
  \begin{align*}
    \alpha_k &= 2\sin \pi \theta -1 \ge 2 \sin (\frac{\pi}{6} - \sqrt{3}\epsilon) -1 \ge 2 \left(\frac{1}{2} \cos (-\sqrt{3}\epsilon) + \frac{\sqrt{3}}{2} \sin (-\sqrt{3}\epsilon)\right) -1\\
    &\ge 2\left(\frac{1}{2} (1 - \frac{2\sqrt{3}}{\pi} \epsilon) - \frac{3 \epsilon}{\pi}\right) -1 \ge -\frac{\epsilon}{\pi}(3 + 2\sqrt{3}) \ge -2\epsilon.
  \end{align*}

  The remaining part of the proposition's statement follows
  immediately from the first part.

  %% Suppose the algorithm returns 0. This implies that we have
  %% $\tilde{\theta} < \frac{1}{6} - \frac{2 \epsilon}{\sqrt{3}\pi}$, so
  %% that $\theta < \frac{1}{6} - \frac{\epsilon}{\sqrt{3}\pi}$ with
  %% probability at least $3/4$ because of \eqref{eq:tildebound}. Thus,
  %% remembering that by assumption we are in the monotone part of the
  %% domain of $\sin$, we have:
  %% \begin{align*}
  %%   \alpha_k = 2 \sin \pi \theta - 1 < 2 \sin \left(\frac{\pi}{6} - \frac{\epsilon}{\sqrt{3}}\right) - 1 < 2 \left(\frac{1}{2} + \frac{\sqrt{3}}{2} \left(\frac{\pi}{6} - \frac{\epsilon}{\sqrt{3}} - \frac{\pi}{6}\right)\right) -1 = -\epsilon,
  %% \end{align*}
  %% where for the second inequality we used the Taylor expansion of $\sin(x)$ at $\frac{\pi}{6}$:
  %% \begin{equation*}
  %%   \sin(x) = \frac{1}{2} + \frac{\sqrt{3}}{2}\left(x - \frac{\pi}{6}\right) - \frac{1}{4}\left(x - \frac{\pi}{6}\right)^2.
  %% \end{equation*}

  %% Finally, suppose $\alpha_k < -2\epsilon$.
 
  Regarding the gate complexity, amplitude estimation with $O(\log
  \frac{1}{\epsilon})$ digits of precision requires
  $O(\frac{1}{\epsilon})$ calls to $U$, controlled-$U$, or
  controlled-$U^{\dag}$, and the reflection circuits of the Grover
  iterator (which can be implemented with $O(q)$ basic gates and
  auxiliary qubits). The inverse quantum Fourier transform on $O(\log
  \frac{1}{\epsilon})$ qubits can be implemented with $O(\log^2
  \frac{1}{\epsilon})$ basic gates; finally, the controlled unitary
  $\ket{0}\bra{0}\otimes U' + \ket{1}\bra{1} \otimes I^{\otimes q}$ to
  construct $\ket{0}\otimes\ket{k}$, as well as the final bitwise
  comparison, can be implemented with $O(q)$ gates. 
  \end{proof}

\begin{algorithm}
  \caption{Sign estimation routine: {\sc SignEstNFP}$(U, k, \epsilon)$.}
  \label{alg:signestnfp}
    \begin{algorithmic}[1]
      \STATE {\bf Input:} state preparation unitary $U$ on $q$ qubits
      (and its controlled version) with $U \ket{0^q} = \sum_{j =
        0}^{2^q - 1} \alpha_j \ket{j}$ and $\alpha_j$ real up to a
      global phase factor, index $k \in \{0,\dots,2^q-1\}$, precision
      $\epsilon$.

      \STATE {\bf Output:} 0 if $\alpha_k \le -\epsilon$, with
      probability at least $3/4$.

      \STATE \label{al:contrupl} Let $V$ map $\ket{0^q}$ to
      $\ket{k}$. Compute $\ket{\psi} = \textsc{Interfere}(U, V)$.

      \STATE Apply amplitude estimation to the state $\ket{0} \otimes
      \ket{k}$ with $\ceil{\log \frac{9 \sqrt{3}\pi}{\epsilon}} + 2$ qubits of
      accuracy; let $\ket{a}$ be the bitstring produced by the phase
      estimation portion of amplitude estimation.

      \STATE If $0.a < \frac{1}{2}$, return 1 if $0.a > \frac{1}{6} -
      \frac{2\epsilon}{3\sqrt{3}\pi}$, 0 otherwise; if $0.a \ge
      \frac{1}{2}$, return 1 if $1 - 0.a > \frac{1}{6} -
      \frac{2\epsilon}{3\sqrt{3}\pi} $, 0 otherwise.
      
    \end{algorithmic}
\end{algorithm}

\begin{proof}{\it Proof of Prop.~\ref{prop:signestnfp}.}
  The proof is similar to Prop.~\ref{prop:signestnfn}, and we use the
  same symbols and terminology.  We apply amplitude estimation to the
  state $\ket{0}\otimes\ket{k}$ to determine the magnitude of
  $\frac{1}{2} (1 + \alpha_k)$. Suppose $\alpha_k \le
  -\epsilon$. Then $\sin \pi \theta \le \frac{1}{2} (1 - \epsilon)$,
  implying:
  \begin{align*}
    \theta &\le \frac{\sin^{-1} \left(\frac{1}{2} (1 - \epsilon)\right)}{\pi} \le
    \frac{\frac{\pi}{6} + \frac{2}{\sqrt{3}}\left(\frac{1}{2} (1 - \epsilon) - \frac{1}{2} \right) + \frac{8}{9\sqrt{3}}\left(\frac{1}{2} (1 - \epsilon) - \frac{1}{2} \right)^2}{\pi} \\
    &\le \frac{\frac{\pi}{6} - \frac{\epsilon}{\sqrt{3}} + \frac{4\epsilon^2}{9\sqrt{3}} }{\pi}
    \le \frac{1}{6} - \frac{7\epsilon}{9\sqrt{3}\pi},
  \end{align*}
  where we used $\epsilon \le \frac{1}{2}$, and for the second inequality we have used the Taylor expansion of $\sin^{-1}(x)$ at $x = \frac{1}{2}$:
  \begin{equation*}
    \sin^{-1}(x) \approx \frac{\pi}{6} + \frac{2}{\sqrt{3}}(x - \frac{1}{2}) + \frac{2\sqrt{3}}{9}(x - \frac{1}{2})^2 + \frac{8\sqrt{3}}{27}(x - \frac{1}{2})^3,
  \end{equation*}
  coupled with the fact that the third-order term is negative at
  $\frac{1}{2} (1 - \epsilon)$. Rather than $\theta$, we obtain an
  approximation $\tilde{\theta}$ up to a certain precision. By
  Thm.~\ref{thm:qpe}, using $\ceil{\log
    \frac{9\sqrt{3}\pi}{\epsilon}} + 2$ qubits of precision, then
  \begin{equation}
    \label{eq:tildebound2}
    |\theta - \tilde{\theta}| \le \frac{\epsilon}{9\sqrt{3}\pi}
  \end{equation}
  with probability at least $3/4$. Then we must have $\tilde{\theta}
  \le \frac{1}{6} - \frac{2\epsilon}{3\sqrt{3}\pi}$ with probability at
  least 3/4. In this case, the algorithm returns 0.

  Now suppose the algorithm returns 0, implying $\tilde{\theta} \le
  \frac{1}{6} - \frac{2\epsilon}{3\sqrt{3}\pi}$. By
  \eqref{eq:tildebound2} we must have so that $\theta \le
  \tilde{\theta} + \frac{\epsilon}{9\sqrt{3}\pi} \le \frac{1}{6} +
  \frac{5\epsilon}{9\sqrt{3}\pi}$. Thus,
  \begin{align*}
    \alpha_k &= 2\sin \pi \theta -1 \le 2\sin (\frac{\pi}{6} +
    \frac{5\epsilon}{9\sqrt{3}}) -1 \le 2 \left(\frac{1}{2} \cos \frac{5\epsilon}{9\sqrt{3}} + \frac{\sqrt{3}}{2} \sin \frac{5\epsilon}{9\sqrt{3}}\right) -1 \\
    &\le \frac{5}{9\sqrt{3}} \epsilon \le \frac{\epsilon}{3}.
  \end{align*}

  The remaining part of the proposition's statement follows
  immediately from the first part, and the running time analysis is
  the same as in Prop.~\ref{prop:signestnfn}.

  %% Now suppose the algorithm returns 1. This implies that we have
  %% $\tilde{\theta} \ge \frac{1}{6} - \frac{2 \epsilon}{3\sqrt{3}\pi}$,
  %% so that $\theta \ge \frac{1}{6} - \frac{3 \epsilon}{9\sqrt{3}\pi}$
  %% with probability at least $3/4$ because of
  %% \eqref{eq:tildebound2}. Thus, for $\theta < \frac{1}{2}$ and
  %% $\epsilon < \frac{1}{2}$, we have:
  %% \begin{align*}
  %%   \alpha_k &= 2 \sin \pi \theta - 1 \ge 2 \sin \left(\frac{\pi}{6} - \frac{3\epsilon}{9\sqrt{3}}\right) - 1 > 2 \left(\frac{1}{2} - \frac{\sqrt{3}}{2} \frac{3\epsilon}{9\sqrt{3}} - \frac{1}{4}\left(\frac{3\epsilon}{9\sqrt{3}}\right)^2 + \frac{\sqrt{3}}{12}\left(\frac{3\epsilon}{9\sqrt{3}}\right)^3 \right) -1 \\
  %%   &\ge -\frac{\epsilon}{3} - \frac{\epsilon}{108} + \frac{\sqrt{3}}{12}\left(\frac{3\epsilon}{9\sqrt{3}}\right)^3 \ge -\epsilon,
  %% \end{align*}
  %% where for the second inequality we used the Taylor expansion of $\sin(x)$ at $\frac{\pi}{6}$:
  %% \begin{equation*}
  %%   \sin(x) = \frac{1}{2} + \frac{\sqrt{3}}{2}\left(x - \frac{\pi}{6}\right) - \frac{1}{4}\left(x - \frac{\pi}{6}\right)^2 - \frac{\sqrt{3}}{12}\left(x - \frac{\pi}{6}\right)^3 + \frac{1}{48}\left(x - \frac{\pi}{6}\right)^4.
  %% \end{equation*}
  \end{proof}

\subsection{Proofs from Section~\ref{sec:pricing}.}

\begin{proof}{\it Proof of Prop.~\ref{prop:redcost}.}
  Let us analyze {\sc RedCost}$(A_B, A_k, c, \epsilon)$. The
  QLSA encodes the solution to:
  \begin{equation}
    \label{eq:redcostls}
    \begin{pmatrix} A_B & 0 \\ 0 & 1
    \end{pmatrix} \begin{pmatrix}x \\ y \end{pmatrix}
    = \begin{pmatrix} A_k \\ c_k \end{pmatrix}    
  \end{equation}
  in a state $\ket{\psi}$ that is guaranteed to be
  $\frac{\epsilon}{10\sqrt{2}}$-close to the exact normalized solution
  $\ket{(A_B^{-1} A_k, c_k)}$. Call this state $\ket{(\tilde{x},
    \tilde{y})}$, where $\tilde{x}, \tilde{y}$ correspond to the
  (approximate) solution of \eqref{eq:redcostls}.

  On line \ref{al:dotprod} we apply the unitary $U_c^{\dag}$,
  obtaining $U_c^{\dag} \ket{(\tilde{x}, \tilde{y})}$. We are now
  interested in tracking the value of the coefficient of the basis
  state $\ket{0^{\ceil{\log m + 1}}}$, which is the input to the {\sc
    SignEstNFN} routine on line \ref{al:signestcall} of {\sc
    CanEnter}. This coefficient is equal to:
  \begin{align*}
    \bra{0^{\ceil{\log m + 1}}} U_c^{\dag}
    \ket{(\tilde{x}, \tilde{y})}) =
    \braket{(-c_B, 1)}{(\tilde{x},\tilde{y})},
  \end{align*}
  because by definition $\bra{0^{\ceil{\log m + 1}}}U_c^{\dag} = (U_c
  \ket{0^{\ceil{\log m +1}}})^{\dag} = \bra{(-c_B, 1)}$. Furthermore,
  \begin{align*}
    \braket{(-c_B, 1)}{(\tilde{x}, \tilde{y})} =
    \frac{1}{\|(-c_B, 1)\|} \left(\tilde{y} - \sum_{j=1}^m c_{B(j)}
    \tilde{x}_j \right).
  \end{align*}
  Again by definition, $\tilde{y}$ is approximately equal to
  $c_k/\|(A_B^{-1} A_k, c_k)\|$ whereas $ \sum_{j=1}^m c_{B(j)}
  \tilde{x}_j$ is approximately equal to $c_B^{\top} A_B^{-1} A_k /
  \|(A_B^{-1} A_k, c_k)\|$. The total error is
  $\frac{\epsilon}{10\sqrt{2}}$, hence, recalling that $\|(-c_B, 1)\| =
  \sqrt{2}$, we have:
  \begin{align*}
    \left| \frac{1}{\sqrt{2}} \left(\tilde{y} - \sum_{j=1}^m c_{B(j)}
    \tilde{x}_j \right) - \frac{1}{\sqrt{2}} \frac{\bar{c}_k}{\|(A_B^{-1} A_k, c_k)\|}
    \right| \le \frac{\epsilon}{10\sqrt{2}}.
  \end{align*}
  This concludes the proof. 
\end{proof}

\begin{proof}{\it Proof of Prop.~\ref{prop:canenter}.}
  {\sc CanEnter} relies on Algorithms \ref{alg:signestnfn} ({\sc
    SignEstNFN}) and \ref{alg:redcost} ({\sc
    RedCost}). 

  On line \ref{al:signestcall} of {\sc CanEnter}, we apply {\sc
    SignEstNFN} to the unitary that implements {\sc RedCost}, with
  $\ket{0^{\ceil{\log m + 1}}}$ as the target basis state and
  precision $\frac{11\epsilon }{10\sqrt{2}}$. Suppose this returns 0. By
  Prop.~\ref{prop:signestnfn} and Prop.~\ref{prop:redcost}, we have:
  \begin{equation*}
    -\frac{11\epsilon}{10\sqrt{2}} > \frac{1}{\sqrt{2}} \left(\tilde{y} - \sum_{j=1}^m c_{B(j)}
    \tilde{x}_j \right) \ge \frac{1}{\sqrt{2}} \frac{\bar{c}_k}{\|(A_B^{-1} A_k, c_k)\|} - \frac{\epsilon}{10\sqrt{2}}.
  \end{equation*}
  This implies:
  \begin{equation*}
    \frac{\bar{c}_k}{\|(A_B^{-1} A_k, c_k)\|} < -\epsilon.
  \end{equation*}

  The above discussion
  guarantee that if the return value of the {\sc CanEnter} is 1, then
  $\bar{c}_k < - \epsilon \|(A_B^{-1} A_k, c_k)\|$,
  with probability at least 3/4, as desired. The gate complexity is
  easily calculated: {\sc SignEst} makes $O(\frac{1}{\epsilon})$ calls
  to {\sc RedCost} and requires an additional $O(\log m + \log^2
  \frac{1}{\epsilon})$ gates. {\sc RedCost} requires one call to the
  QLSA and additional $O(m)$ gates for $U_c$. Thus, the total gate
  complexity is $O(\frac{1}{\epsilon} T_{\text{LS}}(A_B, A_k,
  \frac{\epsilon}{2}) + m + \log^2 \frac{1}{\epsilon})$. 
\end{proof}

\begin{proof}{\it Proof of Thm.~\ref{thm:findcolumn}.}
  {\sc FindColumn} relies on three subroutines: $U_{\text{rhs}}$, {\sc
    CanEnter}, and quantum search, as described in Thm.~\ref{thm:ampamp}.
  $U_{\text{rhs}}$ can be constructed by repeatedly applying the
  procedure of Prop.~\ref{prop:sparsevec} controlled on the register
  containing the column index $k$; the total gate complexity is
  $\tilde{O}(d_c n)$. By Prop.~\ref{prop:canenter}, if the routine
  {\sc CanEnter} returns 1 then the reduced cost of column $A_k$ is $<
  -\epsilon \|(A_B^{-1} A_k, c_k)\|$ with respect to the rescaled
  data, with bounded probability. We can boost this probability with repeated
  applications and a majority vote to make it as close to $1$ as
  desired. Notice that {\sc CanEnter} is applied to the rescaled
  data. In terms of the original data, the condition on the reduced
  cost becomes $\bar{c}_k < -\epsilon \|(A_B^{-1} A_k,
  \frac{c_k}{\|c_B\|})\|$, as claimed in the theorem statement (the
  rescaling of $A_k$ and $A_B^{-1}$ cancel out).

  Finally, we apply the quantum search algorithm
  (Thm.~\ref{thm:ampamp}) using {\sc CanEnter} as the target
  function. The expected number
  of iterations before success is $O(\sqrt{n})$; the gate complexity
  is therefore $\tilde{O}(\frac{\sqrt{n}
  }{\epsilon}(T_{\text{LS}}(A_B, A_k, \frac{\epsilon}{2}) + m))$. By
  Thm.~\ref{thm:qls}, $T_{\text{LS}}(A_B, A_k, \frac{\epsilon}{2})$ is
  $\tilde{O}(\kappa d d_c n + \kappa d^2 m)$, because $P_{A_B}$ has
  gate complexity $\tilde{O}(dm)$ and $P_b$ is the routine
  $U_{\text{rhs}}$, which has gate complexity $\tilde{O}(d_c
  n)$. Thus, we obtain a total gate complexity of
  $\tilde{O}(\frac{\kappa d\sqrt{n}}{\epsilon}( d_c n + d m))$, as
  claimed.  \end{proof}

\begin{proof}{\it Proof of Thm.~\ref{thm:findcolumnse}.}

  The main ingredient of the proof is to analyze the subroutine {\sc
    SteepestEdgeCompare}, Alg.~\ref{alg:secompare}, which is used to
  run the quantum minimum finding algorithm of Thm.~\ref{thm:minfind}.

  At Steps \ref{al:normest1}-\ref{al:normest2} we compute norm
  estimates using Thm.~\ref{thm:qlsnormest}. (In a practical
  implementation, the estimates at Step \ref{al:normest2} could be
  derived from those at Step \ref{al:normest1}, with properly adjusted
  error tolerances.) This requires time
  $\tilde{O}(\frac{1}{\epsilon}T_{\text{LS}}(A_B, A_j, O(\epsilon)))$.

  Recall that $c_{\max} := \max_{\ell \in N} c_\ell$. At Steps
  \ref{al:lsse1}-\ref{al:lsse2} we define two unitaries whose first
  building block is {\sc RedCost}. By
  Prop.~\ref{prop:redcost}, {\sc RedCost} encodes an approximation of
  $\frac{\bar{c}_j}{\|(A_B^{-1}A_j,c_j)\|}$,
  $\frac{\bar{c}_k}{\|(A_B^{-1}A_k,c_k)\|}$ with additive error at
  most $\frac{\epsilon}{10}$ in the amplitude of the all-zero basis
  states; let $\alpha_{0}, \beta_0$ be these amplitudes. We then
  multiply $\alpha_0$ by $\frac{\tilde{d}_j}{\tilde{e}_j c_{\max}}$
  and $\beta_0$ by $\frac{\tilde{d}_k}{\tilde{e}_k c_{\max}}$. Note
  that these terms are $\le 1$ thanks to the factor $c_{\max}$ at the
  denominator, and multiplying the all-zero basis state by such a
  coefficient is simply an application of a unitary with that
  coefficient in the top-left corner. It is known that such a unitary
  can be efficiently constructed with high precision with cost
  $\tilde{O}(1)$; for an explicit construction, see, e.g.,
  \cite[Lemma~48]{gilyen2019quantum}.

  We then apply {\sc Interfere} and {\sc SignEstNFN}, to estimate the
  sign of $\frac{1}{2}(\beta_0 - \alpha_0)$ with error
  $\frac{\epsilon}{8c_{\max}}$. Suppose
  $\frac{\bar{c}_k}{\|A_B^{-1}A_k\|} <
  (1-\epsilon)\frac{\bar{c}_j}{\|A_B^{-1}A_j\|} -
  \epsilon$. We have:
  \begin{align*}
    \beta_0 \le \left(\frac{\bar{c}_k}{\|(A_B^{-1}A_k,c_k)\|} + \frac{\epsilon}{10}\right)\frac{\|(A_B^{-1}A_k,c_k)\|}{c_{\max}\|A_B^{-1}A_k\|}(1+\frac{\epsilon}{2}) < (1-\frac{\epsilon}{2}) \frac{\bar{c}_j}{c_{\max}\|A_B^{-1}A_j\|} - \frac{\epsilon}{c_{\max}} + \frac{\epsilon}{5c_{\max}} \le \alpha_0 - \frac{\epsilon}{4c_{\max}}.
  \end{align*}
  By Prop.~\ref{prop:signestnfn}, this implies that {\sc SignEstNFN}
  returns 0. Thus, applying the quantum minimum finding algorithm of
  Thm.~\ref{thm:minfind} returns an approximate minimizer in
  $O(\sqrt{n})$ iterations, i.e., a column $k$ such that:
  \begin{equation*}
    \frac{\bar{c}_k}{\|A_B^{-1}A_k\|} \le (1+\epsilon)\frac{\bar{c}_h}{\|A_B^{-1}A_j\|} + \epsilon,
  \end{equation*}
  where $h$ is the $\arg \min$ as defined in the theorem statement.

  The running time is obtained by multiplying the gate complexity of
  the {\sc SteepestEdgeCompare} subroutine by the number of iterations
  $O(\sqrt{n})$. In each call to {\sc SteepestEdgeCompare}, we perform
  a constant number of norm estimations using
  Thm.~\ref{thm:qlsnormest}, we apply {\sc RedCost} and {\sc
    SignEstNFN} with precision $O(\frac{\epsilon}{c_{\max}})$. By
  Prop.~\ref{prop:signestnfn} and Prop.~\ref{prop:redcost}, the running
  time is $\tilde{O}(\frac{\kappa c_{\max} d
    \sqrt{n}}{\epsilon}\allowbreak (d_c n + d m))$.   
  %% Suppose {\sc SignEstNFN} returns
  %% 0. Then by Prop.~\ref{prop:signestnfn} we have $\frac{1}{2}(\beta_0
  %% - \alpha_0) < -\frac{\epsilon}{8c_{\max}}$, and the above analysis
  %% implies:
  %% \begin{align*}
  %%   \frac{\|(A_B^{-1}A_k,c_k)\|}{c_{\max}\|A_B^{-1}A_k\|}\left(\frac{\bar{c}_k}{\|(A_B^{-1}A_k,c_k)\|} - \frac{\epsilon}{10}\right) \le \beta_0 < \alpha_0 - \frac{\epsilon}{4c_{\max}} \le
  %%   \frac{\|(A_B^{-1}A_j,c_j)\|}{c_{\max}\|A_B^{-1}A_j\|}\left(\frac{\bar{c}_j}{\|(A_B^{-1}A_j,c_j)\|} + \frac{\epsilon}{10}\right) - \frac{\epsilon}{4c_{\max}}.
  %% \end{align*}
  %% Thus, $\frac{\bar{c}_k}{\|A_B^{-1}A_k\|} <
  %% (1+\epsilon)\frac{\bar{c}_j}{\|A_B^{-1}A_j\|} - \frac{\epsilon}{4}$.
\end{proof}

\subsection{Proofs from Section~\ref{sec:normest}.}
\begin{proof}{\it Proof of Prop.~\ref{prop:normest}.}
  The algorithm works as follows. We apply the unitary
  $U_{\text{rhs}}$ that maps $\ket{0^{\ceil{\log n}}} \otimes
  \ket{0^{\ceil{\log m}}} \to \sum_{k \in N} \frac{\|A_k\|}{\|A_N\|_F}
  \left(\ket{k} \otimes \ket{A_k}\right)$; we call the first register,
  that loops over $k \in N$, the ``column register''. This unitary can
  be constructed with $\tilde{O}(d_c n)$ gates, using
  Prop.~\ref{prop:sparsevec}. (The time to classically compute the
  norms can be ignored, as this only needs to be done once, and its
  time complexity is less than the total complexity of the algorithm.)
  We then apply the QLSA of \cite[Thm.~3]{childs2017quantum}, using
  $U_{\text{rhs}}$ as the oracle $P_{b}$, $P_{A_B}$ as the oracle for the
  constraint matrix, and precision $\frac{\epsilon}{2n}$. As a result,
  conditioned on the column register being $\ket{k}$, we obtain
  $\ket{A_B^{-1} A_k}$ in the output register of the QLSA
  algorithm. Following \cite{childs2017quantum}, there exists an
  auxiliary register that has value $\ket{0^r}$ with amplitude
  $\frac{1}{\alpha}\|\tilde{A}_B^{-1} \ket{A_k}\| =
  \frac{1}{\alpha\|A_k\|}\|\tilde{A}_B^{-1} A_k\|$, where $\alpha$ is
  a known number with $\alpha = O(\kappa
  \sqrt{\log(n\kappa/\epsilon)})$, and $\tilde{A}_B^{-1}$ is an
  operator that is $\frac{\epsilon}{2n}$-close to $A_B^{-1}$ in the
  operator norm. Since this is true for all $k$, the probability of
  obtaining $\ket{0^r}$ in the auxiliary register is:
  \begin{equation}
    \label{eq:amppsucc}
    \sum_{k \in N} \frac{\|A_k\|^2}{\|A_N\|_F^2}
    \frac{1}{\alpha^2\|A_k\|^2}\|\tilde{A}_B^{-1} A_k\|^2 =
    \frac{\|\tilde{A}_B^{-1} A_N\|_F^2}{\alpha^2 \|A_N\|_F^2}.
  \end{equation}
  Using Thm.~\ref{thm:ampest} to estimate this probability, amplitude
  estimation with precision $\frac{\epsilon}{4 \pi \alpha^2}$ yields
  an estimate $\tilde{a}$ of \eqref{eq:amppsucc} with error
  $\pm(\frac{\epsilon}{4 \alpha^2} + \frac{\epsilon^2}{16 \alpha^4})$,
  see \citep[Thm.~12]{brassard2002quantum}. Our estimate for
  $\|A_B^{-1} A_N\|_F^2$ is $\rho : =\tilde{a} \alpha^2 \|A_N\|_F^2$.
  We have:
  \begin{align*}
    \frac{\rho}{\|A_B^{-1} A_N\|_F^2} &\le \frac{\alpha^2
      \|A_N\|_F^2}{\|A_B^{-1} A_N\|_F^2}\left(\frac{\|\tilde{A}_B^{-1}
      A_N\|_F^2}{\alpha^2 \|A_N\|_F^2} + \frac{\epsilon}{4 \alpha^2} +
    \frac{\epsilon^2}{16 \alpha^4}\right) \le 1 + \frac{\epsilon}{2} +
    \frac{\alpha^2 \|A_N\|_F^2}{\|A_B^{-1}
      A_N\|_F^2}\left(\frac{\epsilon}{4 \alpha^2} +
    \frac{\epsilon^2}{16 \alpha^4}\right) \\ &\le 1 + \frac{\epsilon}{2}
    + \alpha^2 \left(\frac{\epsilon}{4 \alpha^2} +
    \frac{\epsilon^2}{16 \alpha^4}\right) \le 1 + \epsilon,
  \end{align*}
  where we have used the fact that $\frac{\|A_N\|_F^2}{\|A_B^{-1}
    A_N\|_F^2} \le 1$ because the smallest singular value of of
  $A_B^{-1}$ is $1$, and the fact that we can assume $\epsilon \le
  \alpha^2$ so that $\frac{\epsilon^2}{16 \alpha^2} \le
  \frac{\epsilon}{16}$. A similar calculation yields the lower bound,
  yielding the desired approximation. Regarding the running time,
  amplitude estimation with precision $\frac{\epsilon}{4 \pi
    \alpha^2}$ requires $\frac{4 \pi \alpha^2}{\epsilon} =
  \tilde{O}(\frac{\kappa^2}{\epsilon})$ applications of the QLSA. The
  running time for the QLSA is $\tilde{O}(\kappa d_c n + \kappa d^2
  m)$, because $P_{A_B}$ has gate complexity $\tilde{O}(dm)$ and $P_b$
  is the routine $U_{\text{rhs}}$, which has gate complexity
  $\tilde{O}(d_c n)$. Thus, we obtain a total running time of
  $\tilde{O}(\frac{\kappa^2}{\epsilon}(\kappa d_c n + \kappa d^2
  m))$. 
\end{proof}

\subsection{Proofs from Section~\ref{sec:ratio_test}.}
\begin{proof}{\it Proof of Prop.~\ref{prop:isunbounded}.}
  Let $\ket{\tilde{x}}$ be the state produced by the QLSA, where
  $\tilde{x}$ is an approximate solution to the linear system; by
  Thm.~\ref{thm:qls}, we have $\|\tilde{x} - \frac{A_B^{-1}
    A_k}{\|A_B^{-1} A_k\|}\| \le \frac{\delta}{10}$.

  Suppose the algorithm returns 1; this implies that the index $\ell$
  is not found. Then if all algorithms are successful, the routine
  {\sc SignEstNFN}${}^+(U_{\text{LS}}, \ell, \frac{9\delta}{10})$ must
  have returned 0 for all $\ell$. by Prop.~\ref{prop:signestnfn},
  $\tilde{x}_{\ell} < \frac{9\delta}{10}$, thus:
  \begin{align*}
    \frac{(A_B^{-1} A_k)_{\ell}}{\|A_B^{-1} A_k\|} \le \tilde{x}_{\ell} + \frac{\delta}{10}  < \frac{9\delta}{10} + \frac{\delta}{10} = \delta,
  \end{align*}
  which implies $(A_B^{-1} A_k)_{\ell} < \delta {\|A_B^{-1}
    A_k\|}$. Thus, if {\sc IsUnbounded} returns 1 and all algorithms
  are successful, it must be that $(A_B^{-1} A_k)_{\ell} <
  \delta{\|A_B^{-1} A_k\|}$ for all $\ell$, which is the condition
  in the statement of the proposition; in this case, the LP is
  indeed unbounded.

  We now analyze the running time.  To determine with constant
  probability if the sought index $\ell$ exists, i.e., if $g(\ell) =
  1$ for some value of $\ell$, we apply amplitude estimation with
  $O(\sqrt{m})$ applications of the Grover operator, see
  Thm.~\ref{thm:ampest}. Each application takes time
  $O(\frac{1}{\delta}(\kappa d^2 m))$, which is the complexity of
  running the QLSA and the sign estimation. The probability of success
  for the QLSA is bounded, and we have a way of determining success,
  therefore we can boost the probability of correctness arbitrarily
  high with enough repetitions of the algorithm. 
\end{proof}

\begin{proof}{\it Proof of Thm.~\ref{thm:qratio_test}.}
  Recall the formula for the ratio test, eq.~\eqref{eq:ratio_test}:
  \begin{equation*}
    r^* := \min_{j=1,\dots,m : u_j > 0} \frac{x_{B(j)}}{u_j},
  \end{equation*}
  where $x_B = A_B^{-1} b, u = A_B^{-1} A_k$. Note that since $x_B \ge
  0$, this can be rewritten as follows:
  \begin{equation}
    r^* := \max \{r : x_B - r u \ge 0 \}.
  \end{equation}
  Thus, an exact solution of the ratio test is attained by any index $j$
  such that $(A_B^{-1} b - r^* A_B^{-1}A_k)_j = (x_B - r^* u)_j = 0$.

  In Alg.~\ref{alg:findrow}, on line~\ref{al:compfind} we define the
  unitary $U_r$. By Prop.~\ref{prop:signestnfn}, with the tolerances set
  in the algorithm, if:
  \begin{equation*}
    \frac{(A_B^{-1} (b - r A_k))_j}{\|A_B^{-1}(b - r A_k)\|} < - \frac{\delta}{4\|A_B^{-1}(b - r A_k)\|},
  \end{equation*}
  then {\sc SignEstNFN} returns 0. Removing the normalization factor
  and taking into account the error of the QLSA, which we set to be at
  most $\frac{\delta}{4\|A_B^{-1}(b - r A_k)\|}$, we have that if
  $(A_B^{-1} (b - r A_k))_j < - \frac{\delta}{2}$ then {\sc
    SignEstNFN} returns 0 with high probability (that can be boosted,
  as usual, by doing multiple repetitions, taking the majority vote,
  and using the Chernoff bound). Similarly, by adjusting the tolerance
  and using the other side of Prop.~\ref{prop:signestnfn}, we can
  define a unitary that returns 1 if $(A_B^{-1} (b - r A_k))_j \ge -
  \frac{\delta}{2}$.

  We now analyze the binary search at line~\ref{al:binsearch}. We can
  compute an upper bound on the maximum value of $r$ by starting with
  $r = 1$ and doubling every time until we find a value such that
  $A_B^{-1} b - r A_B^{-1}A_k \not\ge -\frac{\delta}{2}
  \mathbf{1}_m$; this gives an estimate of the upper range for $r$
  that is off by at most a factor of two. Note that the maximum
  value for $r$ is $O(\max\{\|A_B^{-1}b\|, \|A_B^{-1} A_k\|\})$. We are
  seeking a value $\tilde{r}$ such that $A_B^{-1} b - \tilde{r}
  A_B^{-1}A_k \not\ge -\frac{\delta}{2} \mathbf{1}_m$ but $A_B^{-1}
  b - (\tilde{r}-\frac{\delta}{2\kappa\|A_k\|}) A_B^{-1}A_k \ge
  -\frac{\delta}{2}\mathbf{1}_m$; this can be done by binary search
  in time $O(\log \frac{\kappa\max\{A_B^{-1}b, A_B^{-1}
    A_k\}\|A_k\|}{\delta}) = \tilde{O}(1)$. For the return value
  $\tilde{r}$, for every (basic) index $j$, we have $(A_B^{-1} (b -
  (\tilde{r} - \frac{\delta}{2\kappa\|A_k\|}) A_k))_j \ge -
  \frac{\delta}{2}$, which implies $(A_B^{-1} (b - \tilde{r}
  A_k))_j \ge - \frac{\delta}{2} - \frac{\delta}{2\kappa\|A_k\|}
  (A_B^{-1}A_k)_j \ge - \delta$. This implies that the new basic
  solution obtained by moving by $\tilde{r}$ along the nonbasic edge
  $k$ is $\delta$-feasible, i.e., it violates nonnegativity by at
  most $\delta$. Additionally, we know that there exists $j$ such
  that $(A_B^{-1} (b - \tilde{r} A_k))_j < -\frac{\delta}{2}$, and
  such an index $j$ is precisely the value returned by {\sc FindRow};
  thus, the algorithm returns a row index for the pivot that yields
  the next basic (approximately) feasible solution.

  Regarding the running time, the binary search requires
  $\tilde{O}(1)$ iterations, as stated; at each binary search
  iteration, we need $\tilde{O}(\sqrt{m})$ amplitude estimation
  iterations (Thm.~\ref{thm:ampest}) at line~\ref{al:compfind}, that
  multiplies the cost of running {\sc SignEstNFN} on the solution of
  the QLSA, which is $\tilde{O}(\frac{\|A_B^{-1}A_k\|}{\delta}
  T_{\text{LS}}(A_B, A_k, \frac{\delta}{8\|A_B^{-1}(b - r A_k)\|}))
  = \tilde{O}(\frac{\|A_B^{-1}A_k\|}{\delta} \kappa d^2 m)$. The total running time of the algorithm is
  therefore $\tilde{O}(\frac{1}{\delta} \eta d^2 \kappa^2 m^{1.5})$. 
\end{proof}

\subsection{Proofs from Section~\ref{sec:harris}.}
\begin{proof}{\it Proof of Prop.~\ref{prop:isfeasible}.}
  Let $\ket{\tilde{x}}$ be the state produced by the QLSA; by
  Thm.~\ref{thm:qls}, we have $\|\tilde{x} - \frac{A_B^{-1}
    b}{\|A_B^{-1} b\|}\| \le \frac{\delta}{10}$.

  Suppose $A_B^{-1} b \not\ge -\delta \mathbf{1}_m$, i.e.,
  the basic solution is infeasible. Then, for some index $\ell$, we
  must have $\frac{(A_B^{-1} b)_{\ell}}{\|A_B^{-1} b\|} <
  -\frac{\delta}{\|A_B^{-1} b\|}$. This implies:
  \begin{align*}
    \tilde{x}_{\ell} \le \frac{(A_B^{-1} b)_\ell}{\|A_B^{-1} b\|} +
    \frac{\delta}{10\|A_B^{-1} b\|} < -\frac{\delta}{\|A_B^{-1} b\|} + \frac{\delta}{10\|A_B^{-1} b\|} =
    -\frac{9\delta}{10\|A_B^{-1} b\|}.
  \end{align*}
  By Prop.~\ref{prop:signestnfp}, if the routine {\sc
    SignEstNFN}$(U_{\text{LS}}, \ell, \frac{9\delta}{20\|A_B^{-1} b\|})$ is
  successful it returns zero, so that the function $g(\ell)$ evaluates
  to 1. This implies that if all subroutines are successful, {\sc
    IsFeasible} returns 0, as desired. The running time analysis is
  the same as in Prop.~\ref{prop:isunbounded}. 
\end{proof}

\subsection{Proofs from Section~\ref{sec:qram}}
\begin{proof}{\it Proof of Prop.~\ref{prop:qlsaqram}.}
  By \citep[Lemma 50]{gilyen2019quantum}, we can construct a
  $(\mu(A_B), O(\log n), \epsilon/(\kappa^2 \log^3
  \frac{\kappa}{\epsilon}))$ block encoding for $A_B$ using suitable
  data structures, in time $\tilde{O}(1)$. Using the techniques in
  \citep[Sect.~4.3]{chakraborty2018power}, we can then compute a
  normalized version of $A_B^{-1} \ket{x}$ in time $\tilde{O}(\mu(A_B)
  \kappa)$.

  The QRAM data structures describing $A_B$ and $A_N$ can be prepared
  in the claimed time following \cite{kerenidis2018quantum} and
  Prop.~\ref{prop:sparsevec}. After each iteration of the simplex
  method, we simply need to reindex the structures in memory, i.e.,
  swap one nonbasic column with a basic column, which takes at most
  $\tilde{O}(m)$ operations since each column has size $m$. 
\end{proof}

\end{document}